\documentclass[runningheads]{llncs}
\usepackage[T1]{fontenc}
\usepackage{graphicx}
\usepackage{graphicx} 
\usepackage{xcolor}

\usepackage{wrapfig}
\usepackage{subcaption}

\usepackage{tikz}
\usetikzlibrary{arrows,calc,shapes,shadows,decorations.pathmorphing,decorations.pathreplacing,automata,shapes.multipart,positioning}

\tikzstyle{dist}=[circle,fill,inner sep=1.5pt] 
\tikzstyle{state}=[circle, draw]

\usepackage{nicefrac}

\usepackage{relsize}

\usepackage{booktabs}
\usepackage{tabulary}
\usepackage{multirow}

\usepackage{adjustbox}

\usepackage{multicol}
\usepackage{makecell}

\usepackage[shortlabels]{enumitem}

\usepackage[linesnumbered,ruled,vlined]{algorithm2e}
\SetAlCapSkip{2ex plus 0.5ex minus 1ex}
\SetEndCharOfAlgoLine{}
\SetArgSty{textrm}
\SetKwInput{Input}{Input}                
\SetKwInput{Output}{Output} 

\SetCommentSty{commentstyle}
\SetKw{Break}{break}
\SetKwProg{Type}{type}{}{end}
\SetKwProg{Function}{function}{}{end}
\SetKwFor{ForEach}{foreach}{do}{}
\SetKwFor{WhileTrue}{repeat}{}{end}
\SetKw{Continue}{continue}
\SetKwRepeat{DoWhile}{repeat}{while}
\SetKwRepeat{DoUntil}{repeat}{until}

\usepackage{varioref} 

\usepackage{listings}

\usepackage{ltl}

\usepackage{csquotes}

\usepackage{hyperref}
\hypersetup{%
	colorlinks=true,           
	allcolors=blue!70!black,   
	pdfstartview=Fit,          
	breaklinks=true}
\usepackage[nameinlink,noabbrev,capitalise]{cleveref}

\usepackage{xspace} 
\usepackage{xargs} 

\usepackage{etoolbox}

\usepackage{mdframed}

\usepackage{microtype}

\newcommand\R{\mathbb R}

\newcommand\N{\mathbb N}

\newcommand\Q{\mathbb Q}
\newcommand\Qnn{\mathbb Q _{\geq 0}}

\renewcommand{\epsilon}{\varepsilon}

\DeclareMathOperator*{\opt}{opt}

\newcommand{\Expected}{\mathbb{E}}
\renewcommand{\Pr}{\text{Pr}} 

\newcommand{\Distr}{\textit{Distr}}

\newcommand\mdp{\mathcal{M}}
\newcommand{\states}{S}
\newcommand\Act{\textit{Act}}
\newcommand\Trans{\mathbf{P}}

\newcommand\rew{\mathcal{R}}

\newcommand\indexc [1] [\mdp] {#1_c}

\newcommand\mdptup[1] [\refl]{(#1[S], \allowbreak #1[\Act], \allowbreak #1[\Trans])}

\newcommand\Paths[1][\mdp]{\textit{Paths}^{#1}}
\newcommandx\Pathsfrom[2][1=\mdp,2=\state]{\textit{Paths}^{#1}(#2)}
\newcommand\finPaths[1][\mdp]{\textit{Paths}_{\textit{fin}}^{#1}}
\newcommandx\finPathsfrom[2][1=\mdp, 2=\state]{\textit{Paths}_{\textit{fin}}^{#1}(#2)}
\newcommand{\last}{\textit{last}}

\newcommand{\state}{s}
\newcommand{\action}{\alpha}

\newcommand{\statei}{\state_0}
\newcommand{\statea}{\state_1}
\newcommand{\stateb}{\state_2}

\newcommand\dtmc{\mathcal{D}}
\newcommand\dtmctup[1] [\refl]{(#1[S], \allowbreak #1[\Trans])}

\newcommand\sched{\sigma}

\newcommand{\Scheds}[1][\mdp]{\Sigma^{#1}}
\newcommand{\SchedsMD}[1][\mdp]{\Sigma^{#1}_{\textit{MD}}}

\newcommand{\HyperPCTL}{\textsf{\small HyperPCTL}\xspace} 
\newcommand{\PHL}{\textsf{\small PHL}\xspace}

\newcommand{\comp}{\textsf{\small comp}\xspace}

\renewcommand{\phi}{\varphi}
\renewcommand\implies{\Rightarrow}
\renewcommand\iff{\Leftrightarrow}

\newcommand\varstate{\hat{s}}
\newcommand\varsched{\hat{\sched}}

\newcommand\Finally{\LTLdiamond}

\newcommand{\Prob}{\mathbb{P}}

\newcommand{\RelReach}{\textsf{\small RelReach}\xspace} 
\newcommand{\RelReachMD}{\textsf{\small RelReach$_{\textsf{MD}}$}\xspace} 

\renewcommand{\comp}{\bowtie}
\newcommand{\diseq}{\trianglerighteq}

\newcommand{\init}{\textit{init}}

\newcommand{\Init}{\textit{Init}}

\newcommand{\PTIME}{\textsf{\smaller PTIME}\xspace}
\newcommand{\NP}{\textsf{\smaller NP}\xspace}

\newcommand{\PSPACE}{\textsf{\smaller PSPACE}\xspace}
\newcommand{\EXPTIME}{\textsf{\smaller EXPTIME}\xspace}

\newcommand{\comb}{\mathsf{Comb}}
\newcommand{\relInd}{\mathit{ind}}

\newcommand{\TA}{\textbf{TA}\xspace}
\newcommand{\TAone}{\textbf{TA(1)}\xspace}
\newcommand{\TAtwo}{\textbf{TA(2)}\xspace}

\newcommand{\TS}{\textbf{TS}\xspace}
\newcommand{\PW}{\textbf{PW}\xspace}
\newcommand{\PWone}{\textbf{PW(1)}\xspace}
\newcommand{\PWtwo}{\textbf{PW(2)}\xspace}

\newcommand{\SD}{\textbf{SD}\xspace}
\newcommand{\VN}{\textbf{VN}\xspace}
\newcommand{\RT}{\textbf{RT}\xspace}

\newcommand{\TO}{\textbf{TO}}
\newcommand{\OOM}{\textbf{OOM}}

\newcommand{\HyperProb}{\textsc{HyperProb}\xspace} 
\newcommand{\HyperPaynt}{\textsc{HyperPAYNT}\xspace} 
\newcommand{\storm}{\textsc{Storm}\xspace}
\newcommand{\uppaalsmc}{\textsc{\small UPPAAL-SMC}\xspace}
\newcommand{\PRISM}{\textsc{PRISM}\xspace}

\usepackage{multicol}

\newcommand{\orcid}[1]{\href{https://orcid.org/#1}{\includegraphics[width=3mm]{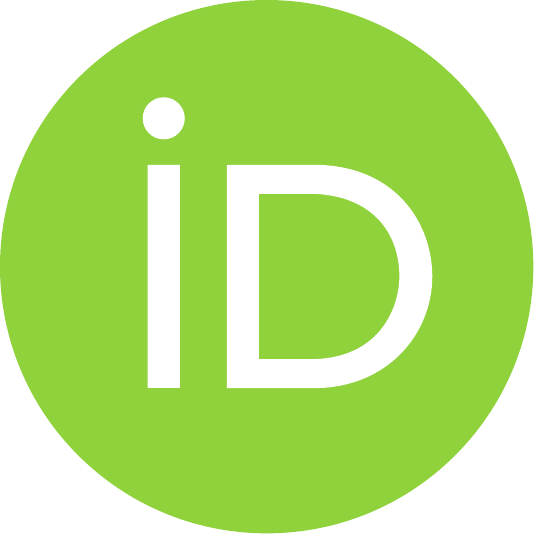}}}

\definecolor{mGreen}{rgb}{0,0.6,0}
\definecolor{mGray}{rgb}{0.5,0.5,0.5}
\definecolor{mPurple}{rgb}{0.58,0,0.82}
\definecolor{backgroundColour}{rgb}{0.95,0.95,0.92}
\lstdefinestyle{CStyle}{
	backgroundcolor=\color{backgroundColour},   
	commentstyle=\color{mGreen},
	keywordstyle=\color{magenta},
	numberstyle=\tiny\color{mGray},
	stringstyle=\color{mPurple},
	basicstyle=\footnotesize,
	breakatwhitespace=false,         
	breaklines=true,                 
	captionpos=b,                    
	keepspaces=true,                 
	numbers=left,                    
	numbersep=2pt,                  
	showspaces=false,                
	showstringspaces=false,
	showtabs=false,                  
	tabsize=1,
	language=C
}

\usepackage[firstpage]{draftwatermark} 
\begin{document}
\title{Efficient Probabilistic Model Checking for Relational Reachability}
%
%
\author{
    Lina Gerlach\inst{1}\orcid{0009-0002-5506-6181} \and 
    Tobias Winkler\inst{1}\orcid{0000-0003-1084-6408} \and 
    Erika {\'A}brah{\'a}m\inst{1}\orcid{0000-0002-5647-6134} 
    \and \\ 
    Borzoo Bonakdarpour\inst{2}\orcid{0000-0003-1800-5419} \and 
    Sebastian~Junges\inst{3}\orcid{0000-0003-0978-8466}
}

\authorrunning{Gerlach et al.}
%
\institute{
    RWTH Aachen University, Aachen, Germany \and 
    Michigan State University, East Lansing, MI, USA \and 
    Radboud University, Nijmegen, the Netherlands
}
\maketitle              

\newtoggle{extended} 
\toggletrue{extended} 

\iftoggle{extended}{
    \SetWatermarkText{} 
}{
    \SetWatermarkAngle{0}
    \SetWatermarkText{\raisebox{12.5cm}{
        \hspace{0.1cm}
        \href{https://doi.org/10.5281/zenodo.15209574}{\includegraphics{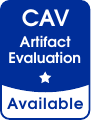}}
        \hspace{9cm}
        \includegraphics{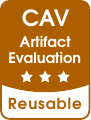}
    }}
}

\begin{abstract}
	
Markov decision processes model systems subject to nondeterministic and 
probabilistic uncertainty. A plethora of verification techniques addresses 
variations of reachability properties, such as: \emph{Is there a 
scheduler resolving the nondeterminism such that the probability to 
reach an error state is above a threshold?} We consider an understudied 
extension that relates different reachability probabilities, such as: \emph{Is there a scheduler such that two sets of states are reached with different probabilities?} These questions appear naturally in the design of randomized 
algorithms and in various security applications. 
We provide a tractable 
algorithm for many variations of this problem, while proving
computational hardness of some others.
An implementation of our algorithm beats solvers for more general probabilistic hyperlogics by orders of magnitude, on the subset of their benchmarks that are within our fragment. 
\end{abstract}
\section{Introduction}
\label{sec:intro}

Markov decision processes (MDPs) are the ubiquitous model to describe system behavior subject to both nondeterminism and probabilistic uncertainty. 
During the execution of an MDP, a \emph{scheduler} (aka \emph{policy} or \emph{strategy}) resolves the nondeterminism at every state by choosing one out of several available actions specifying the possible next states and their probabilities.
A classic verification task in an MDP is to determine whether there is a 
scheduler such that, using this scheduler, the probability to reach an error 
state exceeds a threshold.
To decide this, one can compute and evaluate the performance of the \emph{maximizing} scheduler, as is well established both in theory and in practice~\cite{bk08-book}.
Such \emph{reachability properties} are the cornerstone to support, e.g., linear temporal logics~\cite{DBLP:conf/focs/CourcoubetisY88}.  

However, evaluating the maximizing (or minimizing) scheduler as outlined above is only applicable to the most basic types of reachability properties.
Indeed, for \emph{relational reachability} properties like:
\begin{center}
    \enquote{Is there a scheduler reaching state $a$ with a \emph{different} probability than state $b$?}
\end{center}
we are not necessarily interested in maximizing 
the probability of reaching either $a$ or $b$ --- the two probabilities just have to be different.
\emph{This paper contributes practically efficient algorithms and complexity results for model checking relational reachability properties}, including the above example and many more (see below).

Relational reachability properties go beyond the queries considered in multi-objective model 
checking~\cite{chatterjeeMarkovDecision2006,chatterjeeMarkovDecision2007,etessamiMultiObjectiveModel2008,quatmannVerificationMultiobjective2023}, while they can be expressed in (generally intractable) probabilistic 
hyperlogics~\cite{ab18,dabb22,dft20}.
We discuss the precise relation to these works in \Cref{sec:related}.
Our algorithm solves some standard benchmarks for these hyperlogics  orders of magnitude faster than the state of the art~\cite{dabb21,andriushchenkoDeductiveController2023}.

\paragraph{Motivating example.}
\label{ex:motivating_ex}
We wish to simulate an unbiased, perfectly random coin flip using an infinite stream of possibly \emph{biased} random bits,
each of which is $0$ with an unknown but fixed probability $0<p <1$ and otherwise $1$.
The following simple solution is due to von Neumann~\cite{vonNeumann1951}:
Extract the first two bits from the stream; if they are different, return the value of the first; otherwise try again.
Now, consider a variation of the problem where the stream comprises random bits with 
\emph{different}, unknown biases $p_0, p_1, {\ldots}$ which are, however, all known to lie in an 
interval $[\underline{p}, \overline{p}] \subset (0,1)$.
Is von Neumann's solution still applicable in this new situation?
To address this question for a concrete interval $[\underline{p}, \overline{p}]$, say $[0.59,0.61]$, 
we may model the situation as shown in~\Cref{fig:illustrativeExample} and formalize the property as 
\begin{align}
    \forall \sched .\
    \Pr_{s_0}^{\sched}(\Finally \{01\}) \,\approx_{\epsilon}\, \Pr_{s_0}^{\sched}(\Finally \{10\})
    \tag{$\dagger$}\label{eq:vonNeumannExampleProperty}
\end{align}
where $\sched$ is a universally quantified scheduler, 
$s_0$ is the initial state, and $\approx_{\epsilon}$ means approximate equality up to absolute error $\epsilon \geq 0$.
Using the techniques presented in this paper, we can establish automatically that, as expected, \eqref{eq:vonNeumannExampleProperty} is false for $\epsilon = 0$ (exact equality), but holds if we relax the constraint to $\epsilon = 0.05$.
In words, von Neumann's trick continues to work \enquote{approximately} in the new setting.
Note that the universal quantification in \eqref{eq:vonNeumannExampleProperty} is over \emph{general} policies that may use both unbounded memory and randomization.
This is essential to model the problem properly:
Without randomization, all biases would be either $\underline{p}$ or $\overline{p}$ and a bounded-memory policy would induce an ultimately periodic stream of biases.

\begin{figure}[t]
    \centering
    \scalebox{.7}{
    \begin{tikzpicture}[on grid,node distance=7.5mm and 25mm,semithick,>=stealth]
        \node[state] (start) {$s_0$};
        \node[circle,fill=black,inner sep=1.5pt,above=of start] (startl) {};
        \node[circle,fill=black,inner sep=1.5pt,below=of start] (startu) {};
        \node[state,left=of start] (0) {$0$};
        \node[state,right=of start] (1) {$1$};
        \node[circle,fill=black,inner sep=1.5pt,above=of 0] (0l) {};
        \node[circle,fill=black,inner sep=1.5pt,below=of 0] (0u) {};
        \node[circle,fill=black,inner sep=1.5pt,above=of 1] (1l) {};
        \node[circle,fill=black,inner sep=1.5pt,below=of 1] (1u) {};
        \node[state,accepting,left=of 0] (01) {$01$};
        \node[state,accepting,right=of 1] (10) {$10$};
        
        \draw[-] (start) -- (startl);
        \draw[-] (start) -- (startu);
        \draw[->] (startl) -- node[near start,above] {$0.59$} (0);
        \draw[->] (startl) -- node[near start,above] {$0.41$} (1);
        \draw[->] (startu) -- node[near start,below] {$0.61$} (0);
        \draw[->] (startu) -- node[near start,below] {$0.39$} (1);
        \draw[-] (0) -- (0l);
        \draw[-] (0) -- (0u);
        \draw[-] (1) -- (1l);
        \draw[-] (1) -- (1u);
        \draw[->] (0l) edge[out=0,in=170] node[near start,above] {$0.59$} (start);
        \draw[->] (0l) -- node[near start,above] {$0.41$} (01);
        \draw[->] (0u) edge[out=0,in=190] node[near start,below] {$0.61$} (start);
        \draw[->] (0u) -- node[near start,below] {$0.39$} (01);
        \draw[->] (1l) edge[out=180,in=10] node[near start,above] {$0.41$} (start);
        \draw[->] (1l) -- node[near start,above] {$0.59$} (10);
        \draw[->] (1u) edge[out=180,in=-10] node[near start,below] {$0.39$} (start);
        \draw[->] (1u) -- node[near start,below] {$0.61$} (10);

        \path[-latex', draw]
            (01) edge[loop left] (01)
            (10) edge[loop right] (10);
    \end{tikzpicture}
}
    \caption{MDP from illustrative example. }
    \label{fig:illustrativeExample}
\end{figure}
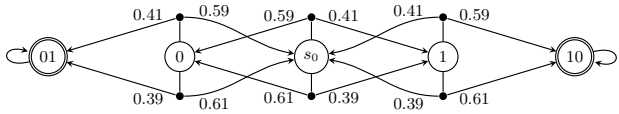

\paragraph{A zoo of relational properties.}
In this paper, we consider variations of relational reachability properties, like
asking for a single 
scheduler ensuring that a set of states is reached with higher probability than a different set of 
states (inequality properties), or that these probabilities are approximately or exactly the same 
(equality properties). These probabilities can be evaluated from the same or from different initial 
states. 
Furthermore, one can consider two schedulers, which allows us to express, e.g., that any pair of 
schedulers induces roughly the same reachability probability. On top of these variations, we 
consider weighted sums of reachability probabilities. We formally introduce the properties in 
\cref{sec:problem_statement}.

\paragraph{Verifying relational reachability properties.}
Given an MDP and a relational reachability property, we provide an algorithm to decide satisfaction and construct, if possible, 
the corresponding witnessing scheduler(s).
The key insight is that these (possibly randomized 
memoryful) schedulers can be constructed by translating the given property to \emph{expected
reward} computations in a series of mildly transformed MDPs.
For inequality properties, it suffices to optimize the total expected reward.
For (approximate) 
equality, the main idea is to construct a randomized memoryful scheduler from the schedulers witnessing the corresponding two inequality properties (\cref{sec:algo}).
Our prototypical implementation of this approach on top of \storm \cite{henselProbabilisticModel2022} shows its practical feasibility (\Cref{sec:evaluation}). 

\paragraph{Computational complexity.}
The algorithm outlined above is exponential \emph{only} in the number of 
different target sets that occur in the property, i.e., the algorithm is \emph{fixed-parameter tractable} (FPT)~\cite{groheDescriptiveParameterized1999}.
In \Cref{sec:complexity}, we study the computational complexity of the problem in greater detail.
Besides the FPT 
result, we show that the problem is in general \PSPACE-hard, but can be solved in \PTIME under one of various (mild) assumptions. Furthermore, when restricting the schedulers to be \emph{memoryless and deterministic} (MD), several types of equality properties are strongly \NP-hard.
We also list various fragments where we can compute MD schedulers in polynomial 
time. 
\cref{tab:complexity_overview_m=2} gives an overview on the complexity of \emph{selected} fragments that compare two probabilities, illustrating the border between \PTIME and strong \NP-hardness for MD schedulers.

\begin{table}[t]
    \caption{Complexity of selected classes of simple relational reachability properties over MD schedulers, where $\epsilon > 0$ and $\diseq~ \in \{ \geq, >, \not\approx_{\epsilon'} \mid \epsilon' \geq 0 \}$.
    Over general schedulers, all variants considered here are in \PTIME (\Cref{th:general_PTIME}).
    }
    \label{tab:complexity_overview_m=2}
    \setlength\tabcolsep{0pt}
    \centering
    \begin{tabular*}{\linewidth}{@{\extracolsep{\fill}} l  l }
        \toprule
         \bf Property class & \bf  Complexity over MD schedulers 
         \\
        \midrule
         $\exists \sched .\ \Pr^{\sched}_{\state}(\Finally T_1) =
       \Pr^{\sched}_{\state}(\Finally T_2)$
        & strongly \NP-complete [Th.~\ref{th:MD_NP-complete}\ref{item:1sched1state}] 
        \\
        $\exists \sched .\ \Pr^{\sched}_{\state}(\Finally T_1) \approx_\epsilon 
       \Pr^{\sched}_{\state}(\Finally T_2)$
        & \NP-complete [Th.~\ref{th:MD_NP-complete}\ref{item:1sched1state}] 
        \\
        $\exists \sched .\ \Pr^{\sched}_{\state}(\Finally T_1) \diseq \Pr^{\sched}_{\state}(\Finally T_2)$
        & in \NP [Th.~\ref{th:MD_NP}]; \PTIME if $T_1, T_2$ absorb. [Th.~\ref{th:MD_PTIME}\ref{item:MD_n-comb}]
        \\
        \midrule 
        $\exists \sched .\ \Pr^{\sched}_{\state_1}(\Finally T_1) = \Pr^{\sched}_{\state_2}(\Finally T_2)$
        & strongly \NP-complete [Th.~\ref{th:MD_NP-complete}\ref{item:1sched2state}]
        \\
        $\exists \sched .\ \Pr^{\sched}_{\state_1}(\Finally T_1) \approx_\epsilon \Pr^{\sched}_{\state_2}(\Finally T_2)$
        & \NP-complete [Th.~\ref{th:MD_NP-complete}\ref{item:1sched2state}]
        \\
        $\exists \sched .\ \Pr^{\sched}_{\state_1}(\Finally T_1) \diseq \Pr^{\sched}_{\state_2}(\Finally T_2)$
        & in \NP [Th.~\ref{th:MD_NP}]
        \\
        \midrule 
        $\exists \sched_1, \sched_2 .\ \Pr^{\sched_1}_{\state_1}(\Finally T_1) = \Pr^{\sched_2}_{\state_2}(\Finally T_2)$
        & strongly \NP-complete [Th.~\ref{th:MD_NP-complete}\ref{item:2sched2state}]
        \\
        $\exists \sched_1, \sched_2 .\ \Pr^{\sched_1}_{\state_1}(\Finally T_1) \approx_\epsilon \Pr^{\sched_2}_{\state_2}(\Finally T_2)$
        & \NP-complete [Th.~\ref{th:MD_NP-complete}\ref{item:2sched2state}]
        \\
        $\exists \sched_1, \sched_2 .\ \Pr^{\sched_1}_{\state_1}(\Finally T_1) \diseq \Pr^{\sched_1}_{\state_2}(\Finally T_2)$
        & \PTIME [Th.~\ref{th:general_PTIME}\ref{item:independent}]
        \\
        \bottomrule
    \end{tabular*}
\end{table}

\paragraph{Summary and contributions.}
In summary, in this paper we present model checking for relational reachability properties, which go 
beyond standard reachability and multi-objective properties while remaining tractable 
(see~\Cref{sec:related} for a discussion of related work). The tractability is in sharp contrast with the 
more general probabilistic hyperlogics. 
The key contributions are the efficient algorithm (\cref{sec:algo}), a 
prototypical implementation thereof (\cref{sec:evaluation}), and a study of the complexity landscape for relational 
reachability properties (\cref{sec:complexity}). 
\iftoggle{extended}{}{
For details on the algorithms, proofs and benchmarks we refer to the extended version~\cite{extendedVersion}.}

\section{Preliminaries}
\label{sec:prelim}

We use $\N$, $\Q$, $\Qnn$, and $\R$ to denote the sets of natural, rational, non-negative rational and real numbers, respectively. 
For $r,r', \epsilon \in \R$ with $\epsilon \geq 0$, we write $r \approx_\epsilon r'$ iff $|r-r'| \leq \epsilon$, and $r \not\approx_\epsilon r'$ iff $|r - r'| > \epsilon$.
The set $\Distr(V)$ of \emph{probability distributions} over a finite set $V$ contains all $\mu \colon V \rightarrow [0,1]$ s.t.\ $ \sum_{v \in V}\mu(v) = 1$.

\begin{definition}
    \label{def:MDP}
    A \emph{Markov decision process (MDP)} is a triple $\mdp {=} \mdptup$ s.t.\ 
    $\states$ is a non-empty finite set of \emph{states},
    $\Act$ is a non-empty finite set of \emph{actions}, and
    $\Trans \colon \states \times \Act \times \states \rightarrow [0,1]$ is a \emph{transition 
    probability function} s.t.\ for all $\state \in \states$ the set of its \emph{enabled actions}
	$
		\Act(\state)= \left\{\action\in\Act \mid \sum_{\state' \in \states} \Trans(\state, 
		\action, \state') =1\right\}
	$
	is non-empty and 
	$\sum_{\state' \in \states} \Trans(\state, \action, \state') = 0$
 for all $\action \in \Act \setminus \Act(s)$.
\end{definition}
A state $\state \in \states$ is \emph{absorbing} if $\Trans(\state, \action, \state) = 1$ for all $\action \in \Act(\state)$.
A set of states $S' \subseteq \states$ is absorbing if all $\state \in S'$ are absorbing.

An \emph{(infinite) path} of an MDP $\mdp$ is an infinite sequence of states and actions $\pi = s_0 \action_0 s_1 \action_1 \ldots$ such that for all $i$ we have 
$\Trans(s_i, \action_i, s_{i+1})>0$.
A \emph{finite path} is a finite prefix of an infinite path ending in a state.
We use $\Paths$ (respectively, $\finPaths$) to denote the set of all infinite (respectively, finite) paths of $\mdp$, and for some state $\state \in \states$ we use $\Pathsfrom$ (respectively, $\finPathsfrom$) to denote the set of all infinite (respectively, finite) paths of $\mdp$ starting at $\state$.
For a finite or infinite path $\pi$ we use $\pi(i) := s_i$ to denote the $i^{th}$ state of $\pi$.
For a finite path $\pi = s_0 \action_0 \ldots \action_{n-1} s_n$, we define $\last(\pi) := s_n$ and $|\pi| := n$.

A \emph{scheduler} resolves the nondeterminism in an MDP.
\begin{definition}
    \label{def:MDPsched}
    A \emph{scheduler} for an MDP $\mdp = \mdptup$ is a function $\sched \colon \finPaths \allowbreak\to\Distr(\Act)$ with $\sched(\pi)(\action) = 0$ for all $\pi \in \finPaths$ and $\action \in \Act \setminus \Act(\last(\pi))$.
\end{definition}

A scheduler is \emph{memoryless} if it can be defined as a function $\sched \colon \states \to 
\Distr(\Act)$, and \emph{memoryful} (oder history-dependent)
otherwise.
A scheduler is \emph{deterministic} if $\sched(\pi) \in \{0,1\}$ for all $\pi \in \finPaths$, and \emph{randomized} otherwise.
The set of all general (i.e., history-dependent randomized (HR)) schedulers for an MDP $\mdp$ is denoted by $\Scheds$, the set of all memoryless deterministic (MD) ones by $\SchedsMD$.

Applying a scheduler to an MDP induces a 
\emph{discrete-time Markov chain (DTMC)}, which is an MDP where the set of actions is a singleton. 
We usually omit the actions and define a DTMC as a tuple $\dtmc = \dtmctup$.

\begin{definition}
    For an MDP $\mdp = \mdptup$ and a scheduler $\sched \in \Scheds$, 
    the \emph{DTMC induced by $\mdp$ and $\sched$} is defined as $\mdp^{\sched} = (\states^{\sched}, \Trans^{\sched})$ with $\states^{\sched} = \finPaths$,
    \[\Trans(\pi, \pi') = \begin{cases}
        \Trans(\last(\pi), \action, s') \cdot \sched(\pi)(\action) & \text{if } \pi' = \pi\action s' \\
        0 & \text{otherwise}\ .
    \end{cases}\]
\end{definition}

For an MDP $\mdp$, a scheduler $\sched$ and a state $\state \in \states^\sched$ in the induced DTMC $\mdp^\sched$, we use $\Pr^{\mdp, \sched}_{\state}$ or simply $\Pr^\sched_\state$ to denote the associated probability measure.
For a target set $T \subseteq \states$, we further use $\Pr^\sched_\state(\Finally T)$ to denote the probability of reaching $T$ from $\state$ in $\mdp^\sched$. 
We refer to~\cite{bk08-book} for details.

\section{Problem Statement}
\label{sec:problem_statement}

The problem we study in this paper is formally defined as follows:

\begin{mdframed}[backgroundcolor=white,linewidth=0.75pt]
    \textbf{Problem \RelReach:} 
    Given an MDP $\mdp$ with states $\states$, decide whether
    \[
        \exists \sched_1, \ldots, \sched_n \in \Scheds .~
        \sum_{i=1}^{m} q_i \cdot \Pr^{\sched_{k_i}}_{\state_{i}}(\Finally T_i) ~\comp~ q_{m+1}
        ~,
        \qquad\text{where}
    \]
    
    \begin{itemize}
        \item $m,n$ are natural numbers with $m \geq n$,
        \item $q_1, \ldots, q_{m+1}$ are rational coefficients,
        \item $\state_1, \ldots, \state_m \in \states$ are (not necessarily distinct) initial states,
        \item $\{k_1, \ldots, k_m\} = \{1, \ldots, n\}$ is a set of indices,
        \item $T_1, \ldots, T_m \subseteq \states$ are (not necessarily distinct) target sets, and 
        \item $\comp~ \in \{ >, \geq, \approx_\epsilon, \not\approx_\epsilon 
        \mid \epsilon \in \Qnn\}$ is a comparison operator.
    \end{itemize}
\end{mdframed}

The comparison operators $=$ and $\neq$ are supported via $\approx_0$ and $\not\approx_0$, respectively, while properties with $\comp~\in\{<,\leq\}$ can be reduced to the above form by multiplying coefficients with $-1$.
Purely universally quantified properties (e.g.\ \eqref{eq:vonNeumannExampleProperty} on page~\pageref{eq:vonNeumannExampleProperty}) are readily reducible to \RelReach via negation.
We do not consider properties with quantifier alternations.
Throughout the paper, the universally quantified variant is loosely referred to as \enquote{relational reachability} as well, but ``\RelReach'' is reserved for the existential variant defined above (this distinction is relevant for the complexity results).
Note that properties like $\exists \sched .\ \Pr^\sched_s(\Finally T) = \Pr^\sched_s(\Finally T)$ can be brought into the above form by subtracting the right-hand-side on both sides of the equality, since we allow positive and negative coefficients.
Below, we provide some further examples properties:
\begin{enumerate}[(1)]
    \item \emph{The probability of reaching $T$ from $s$ is \enquote{approximately scheduler-independent}:}
    \[\forall \sigma_1 \forall \sigma_2.~\Pr_s^{\sigma_1}(\Finally T) \approx_\epsilon \Pr_s^{\sigma_2}(\Finally T) ~.\]
    \item \emph{The probability to reach $T$ from $s_1$ is at least twice \textnormal{/} 10\% higher than the probability of reaching $T$ from $s_2$, no matter the scheduler:}
    \[
        \forall \sigma .~ \Pr_{s_1}^{\sigma}(\Finally T) \geq 2\cdot \Pr_{s_2}^{\sigma}(\Finally T)
        \quad\text{/}\quad
        \forall \sigma .~ \Pr_{s_1}^{\sigma}(\Finally T) \geq \Pr_{s_2}^{\sigma}(\Finally T) + 0.1
        ~.
    \]
    \item \emph{There is a scheduler that, \emph{in expectation}, visits more (different) targets from $\{T_1,\ldots,T_k\}$ than from $\{U_1,\ldots,U_\ell\}$:}
    \[\exists \sigma.~\Pr_s^\sigma(\Finally T_1) + \ldots + \Pr_s^\sigma(\Finally T_k) > \Pr_s^\sigma(\Finally U_1) + \ldots + \Pr_s^\sigma(\Finally U_\ell)~.\]
\end{enumerate}

\section{Verifying Relational Reachability Properties Efficiently}
\label{sec:algo}

\newcommand{\attractor}{A}
\newcommand{\lb}[1]{\underline{#1}}
\newcommand{\ub}[1]{\overline{#1}}

Assume an arbitrary MDP $\mdp = \mdptup$ and a \RelReach property 
\begin{align*}
    \exists \sched_1, \ldots, \sched_n \in \Scheds .~ \sum_{i=1}^{m} q_i \cdot \Pr^{\sched_{k_i}}_{\state_{i}}(\Finally T_i) ~\comp~ q_{m+1} ~.
    \tag{$\star$} \label{eq:genRelReachProp}
\end{align*}
In the following, we outline a four-step procedure that checks whether property (\ref{eq:genRelReachProp}) holds and, if yes, constructs (possibly memoryful randomized) witness schedulers.
\iftoggle{extended}{
    The procedure is summarized in \cref{alg:linear_general}. 
}{
    The procedure is summarized in \cref{alg:linear_general_simplified} for the comparison relation $\approx_\epsilon$. 
}

\begin{example}
    \label{ex:runningExIntro}
    The MDP in \Cref{fig:runningEx} (left) together with the property
    \[
        \exists \sched .~~ \Pr^{\sched}_{s_1}(\Finally T_1) -\nicefrac{1}{2} \cdot \Pr^{\sched}_{s_1}(\Finally T_2) - \nicefrac{1}{2} \cdot \Pr^{\sched}_{s_2}(\Finally T_2) ~\approx_\epsilon~ 0
        ~,
    \]
    (\emph{Does there exist a scheduler such that the probability of reaching $T_1$ from $s_1$ is approximately equal to the mean of the probabilities of reaching $T_2$ from $s_1$ and $T_2$ from $s_2$?}) will serve as a running example throughout the section.
\end{example}

\subsubsection{Step 1: Collect combinations of initial states and schedulers.}
\label{step1}
We start by analyzing the relationship of schedulers and states in the property. 

\begin{definition}[State-scheduler combinations]
    We define 
    $\comb = \{ (s_i, \sched_{k_i}) \mid i=1, \ldots, m\}$, the set of all different combinations of initial states and schedulers that occur in the property \eqref{eq:genRelReachProp}.
    Furthermore, for every $c = (s, \sched) \in \comb$, we define $\relInd(c)$ as the set of indices $i$ such that $(s_i, \sched_{k_i}) = c$.
\end{definition}

\begin{example}
    In the property from \Cref{ex:runningExIntro} we have $\comb = \{c_1=(s_1,\sched), c_2=(s_2,\sched)\}$.
    Furthermore, $\relInd(c_1) = \{1,2\}$ and $\relInd(c_2) = \{3\}$. 
\end{example}

Notice that $n \leq |\comb| \leq m$.
State-scheduler combinations allow introducing fresh scheduler variables, one per combination:
\begin{lemma}
    \label{thm:whyComb}
    Let $\comb = \{c_1,\ldots,c_{k}\}$.
    Then property \eqref{eq:genRelReachProp} is equivalent to
    \[
        \exists \sched_{c_1},\ldots,\sched_{c_k}.~ \sum\nolimits_{i=1}^{k}  \Big[ \sum\nolimits_{j \in \relInd(c_i)} q_j \cdot \Pr_{s_j}^{\sched_{c_i}}(\Finally T_j) \Big]
        ~\comp~ q_{m+1}
        ~.
    \]
\end{lemma}
\iftoggle{extended}{
    \begin{proof}
    Quantifying over each state-scheduler combination individually is justified because of the following:
    For every pair of distinct combinations $c = (s,\sched), c' = (s',\sched') \in \comb$ it holds that either the scheduler variables are already distinct (i.e., $\sched\neq\sched'$), or else $\sched = \sched'$ and $s \neq s'$.
    In the latter case, since \emph{memoryful schedulers may depend on the initial state}, we can instead quantify over two different schedulers.
    More formally, the following is true in any MDP $\mdp$ with states $s \neq s'$ (and target sets $T, T'$):
    For every $\sched,\sched' \in \Scheds[\mdp]$ there exists a (memoryful) $\hat{\sched}\in \Scheds[\mdp]$  such that 
    $\Pr^{\sched}_{s}(\Finally T) = \Pr^{\hat{\sched}}_{s}(\Finally T)$ and $\Pr^{\sched'}_{s'}(\Finally T') = \Pr^{\hat{\sched}}_{s'}(\Finally T')$.
    \hfill\qed
\end{proof}
}{
    \begin{proof}[Sketch]
        Quantifying over each state-scheduler combination individually is justified because schedulers may use memory and thus remember the initial state, see \cite{extendedVersion} for details. 
        \hfill\qed
    \end{proof}
}

\begin{example}
    Applying \Cref{thm:whyComb} to the property from \Cref{ex:runningExIntro} yields the following equivalent property (over general, memoryful schedulers):
    \[
        \exists \sched_{c_1}, \sched_{c_2} .~
        \underbrace{\left[1 \cdot \Pr^{\sched_{c_1}}_{s_1}(\Finally T_1) -\tfrac{1}{2} \cdot \Pr^{\sched_{c_1}}_{s_1}(\Finally T_2)\right]}_{\text{combination } c_1 \,=\, (s_1,\textcolor{gray}{\sched})}
        +
        \underbrace{\left[-\tfrac{1}{2} \cdot \Pr^{\sched_{c_2}}_{s_2}(\Finally T_2)\right]}_{\text{combination } c_2 \,=\, (s_2,\textcolor{gray}{\sched})}
        ~\approx_\epsilon~
        0
    ~.
    \]
\end{example}

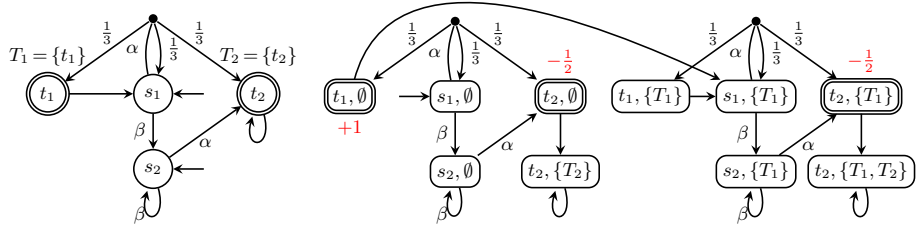
\begin{figure}[t]
    \centering
    \begin{minipage}[t]{0.34\textwidth}
        \begin{tikzpicture}[initial text=,initial where=right,node distance=10mm and 14mm,on grid,semithick,>=stealth, every node/.style={scale=0.8}]
            \node[state,initial] (s1) {$s_1$};
            \node[state,initial,below=of s1] (s2) {$s_2$};
            \node[state,accepting,right=of s1,label=90:{$T_2\,{=}\,\{t_2\}$}] (t2) {$t_2$};
            \node[dist,above =of s1] (a) {};
            \node[state,accepting,left=of s1,label=90:{$T_1\,{=}\,\{t_1\}$}] (t1) {$t_1$};
            \draw[->] (s1) edge node[left] {$\beta$} (s2);
            \draw[-] (s1) edge[bend left=15] node[left] {$\alpha$} (a);
            \draw[->] (s2) edge node[below] {$\alpha$} (t2);
            \draw[->] (s2) edge[loop below] node[left] {$\beta$} (s2);
            \draw[->] (t2) edge[loop below] (t2);
            \draw[->] (t1) edge (s1);
            \draw[->] (a) edge[bend left=15] node[right] {$\tfrac 1 3$} (s1);
            \draw[->] (a) edge node[above] {$\tfrac 1 3$}(t1);
            \draw[->] (a) edge node[above] {$\tfrac 1 3$} (t2);
        \end{tikzpicture}
    \end{minipage}
    \hfill
    \begin{minipage}[t]{0.65\textwidth}
        \begin{tikzpicture}[initial text=,initial where=left,node distance=10mm and 14mm,on grid,semithick,>=stealth, every node/.style={scale=0.8}]
            \node[state,rectangle,rounded corners,initial] (s1) {$s_1,\emptyset$};
            \node[state,rectangle,rounded corners,below=of s1] (s2) {$s_2,\emptyset$};
            \node[state,rectangle,rounded corners,accepting,right=of s1,label={90:\textcolor{red}{$-\frac 1 2$}}] (t2) {$t_2,\emptyset$};
            \node[dist,above =of s1] (a) {};
            \node[state,rectangle,rounded corners,accepting,left=of s1,label={-90:\textcolor{red}{$+1$}}] (t1) {$t_1,\emptyset$};
            \node[state,rectangle,rounded corners,below=of t2] (t2new) {$t_2,\{T_2\}$};
            
            \node[state,rectangle,rounded corners,right=40mm of s1] (s1') {$s_1,\{T_1\}$};
            \node[state,rectangle,rounded corners,below=of s1'] (s2') {$s_2,\{T_1\}$};
            \node[state,rectangle,rounded corners,accepting,right=of s1',label={90:\textcolor{red}{$-\frac 1 2$}}] (t2') {$t_2,\{T_1\}$};
            \node[dist,above =of s1'] (a') {};
            \node[state,rectangle,rounded corners,left=of s1'] (t1') {$t_1,\{T_1\}$};
            \node[state,rectangle,rounded corners,below=of t2'] (t2new') {$t_2,\{T_1,T_2\}$};
            
            \draw[->] (s1) edge node[left] {$\beta$} (s2);
            \draw[-] (s1) edge[bend left=15] node[left] {$\alpha$} (a);
            \draw[->] (s2) edge node[below] {$\alpha$} (t2);
            \draw[->] (s2) edge[loop below] node[left] {$\beta$} (s2);
            \draw[->] (t2) edge (t2new);
            \draw[->] (t2new) edge[loop below] (t2new);
            \draw[->] (t1) edge[out=75,in=155] (s1');
            \draw[->] (a) edge[bend left=15] node[right] {$\tfrac 1 3$} (s1);
            \draw[->] (a) edge node[above] {$\tfrac 1 3$}(t1);
            \draw[->] (a) edge node[above] {$\tfrac 1 3$} (t2);
            
            \draw[->] (s1') edge node[left] {$\beta$} (s2');
            \draw[-] (s1') edge[bend left=15] node[left] {$\alpha$} (a');
            \draw[->] (s2') edge node[below] {$\alpha$} (t2');
            \draw[->] (s2') edge[loop below] node[left] {$\beta$} (s2');
            \draw[->] (t2') edge (t2new');
            \draw[->] (t2new') edge[loop below] (t2new');
            \draw[->] (t1') edge (s1');
            \draw[->] (a') edge[bend left=15] node[right] {$\tfrac 1 3$} (s1');
            \draw[->] (a') edge node[above] {$\tfrac 1 3$}(t1');
            \draw[->] (a') edge node[above] {$\tfrac 1 3$} (t2');
        \end{tikzpicture}
    \end{minipage}
    \caption{An MDP (left) and its goal unfolding with \textcolor{red}{rewards} (right).}
    \label{fig:runningEx}
\end{figure}

\subsubsection{Step 2: Unfold targets and set up reward structures.}
\label{step2}
Next we process each combination $c \in \comb$ individually.
We rely on two established techniques from the literature:
Including reachability information in the state space \cite{quatmannVerificationMultiobjective2023,forejtQuantitativeMultiobjective2011} 
and encoding reachability probabilities as \emph{expected rewards} (e.g.,~\cite[pp.~51~ff.]{quatmannVerificationMultiobjective2023}).
For the sake of completeness, we detail these steps nonetheless:

\begin{definition}[Goal unfolding]
    \label{def:goalUnfolding}
    Let $c \in \comb$.
    The \emph{goal unfolding} of $\mdp$ w.r.t.\ $c$ is the MDP $\indexc = (\states_c, \Act, \indexc[\Trans])$ where $\states_c = \states \times 2^{\mathcal{T}_c}$ for $\mathcal{T}_c = \{T_i \mid i \in \relInd(c)\} \neq\emptyset$ the target sets corresponding to $c$, and
    $\indexc[\Trans]$ is defined as follows:
    For $\state, \state' \in \states$, $\mathcal{T}, \mathcal{T}' \subseteq \mathcal{T}_c$, and $\action \in \Act$,
    \[
        \indexc[\Trans]\big((\state, \mathcal T), \action, (\state', \mathcal T')\big)
        ~=~ 
        \begin{cases}
            \Trans(\state, \action, \state') & \text{if } 
            \mathcal T' = \mathcal T \cup \mathcal \{ T \in \mathcal{T}_c \mid s \in T \} \\
        0 & \text{else}
        \end{cases}
        ~.
    \]
\end{definition}
For a combination $c = (s,\textcolor{gray}{\sched})$ we use $\state_{c}$ to denote the state $(\state, \emptyset)$ in $\mdp_c$.

\begin{example}
    \label{ex:unfolding}
    The goal unfolding of the MDP in \Cref{fig:runningEx} (left) w.r.t.\ combination $c_1 = (s_1,\textcolor{gray}{\sched})$, for which we have $\mathcal T_{c_1} = \{T_1,T_2\}$, is depicted in \Cref{fig:runningEx} (right).
\end{example}

\begin{definition}[Reward structure for state-scheduler combination]
    \label{def:rewForComb}
    Let $c \in \comb$.
    We define the reward structure $\indexc[\rew] \colon \states_c \to \Q$ on the goal unfolding $\indexc$ by $\indexc[\rew] = \sum_{T \in \mathcal T_c} q_T \cdot \rew_{T}$, 
    where $q_T = \sum_{i \in \{\relInd(c)\mid T = T_i\}} q_i$ and
    \[
        \rew_T \colon \indexc[\states] \to \Q,~
        (\state, \mathcal T) \mapsto
        \begin{cases}
            1 & \text{if } \state \in T \wedge T \not\in \mathcal T \\
            0 & \text{else}
        \end{cases}
        ~.
    \] 
\end{definition}
Intuitively, we collect reward equal to the sum of the coefficients occurring together with a target $T \in \mathcal T_c$ when we visit $T$ \emph{for the first time}.
For any $\sched \in \Scheds[\indexc]$, the reward function $\indexc[\rew]$ can be naturally lifted to $\indexc^\sched$ and further to infinite paths of $\indexc^\sched$ by letting $\indexc[\rew](\pi) = \sum_{i=0}^{\infty} \indexc[\rew](\pi(i))$ for $\pi \in \Paths[\indexc]$; this is well-defined because we collect reward only finitely often on any path.
The \emph{expected reward} of $\indexc[\rew]$ on $\indexc$ from $s_c$ under some $\sched \in \Scheds[\indexc]$ is then defined as the expectation of the function $\indexc[\rew](\pi) = \sum_{i=0}^{\infty} \indexc[\rew](\pi(i))$.
Then, we can reduce our query to a number of expected reward queries\iftoggle{extended}{}{ (see \cite{extendedVersion} for the proof)}:

\begin{lemma}
    \label{thm:weightedReachProbsAsExpectedRew}
    For every combination $c = (s, \textcolor{gray}{\sched}) \in \comb$ and $\opt 
    \in \{\min,\max\}$:
    \[
        \opt_{\sched \in \Scheds[\mdp]} \sum_{j \in \relInd(c)} q_j \cdot \Pr_{s}^{\mdp,\sched}(\Finally T_j)
        ~=~
        \opt_{\sched \in \Scheds[\indexc]} \Expected^{\indexc, \sched}_{s_c}(\indexc[\rew])
        ~.
    \]
\end{lemma}

\iftoggle{extended}{\begin{proof}
    Let $c = (s, \textcolor{gray}{\sched}) \in \comb$ and $\opt 
    \in \{\min,\max\}$.
    By construction, we have
    \[
        \opt_{\sched \in \Scheds[\mdp]} \sum_{j \in \relInd(c)} q_j \cdot \Pr_{s}^{\mdp,\sched}(\Finally T_j) 
        = 
        \opt_{\sched \in \Scheds[\indexc]} \sum_{j \in \relInd(c)} q_j \cdot \Pr_{s_c}^{\indexc,\sched}(\Finally T_j) 
        ~.
    \]
    Since $\rew_{T_j}$ collects reward 1 only on the first visit to $T_j$, it further follows that 
    \[
        \opt_{\sched \in \Scheds[\indexc]} \sum_{j \in \relInd(c)} q_j \cdot \Pr_{s_c}^{\indexc,\sched}(\Finally T_j)
        = 
        \opt_{\sched \in \Scheds[\indexc]} \sum_{j \in \relInd(c)} q_j \cdot \Expected^{\indexc, \sched}_{s_c}(\rew_{T_j}) 
        ~.
    \]
    By linearity of expectations, it holds that
    \begin{align*}
        \opt_{\sched \in \Scheds[\indexc]} \sum_{j \in \relInd(c)} q_j \cdot \Expected^{\indexc, \sched}_{s_c}(\rew_{T_j}) 
        &= 
        \opt_{\sched \in \Scheds[\indexc]} \Expected^{\indexc, \sched}_{s_c}\left(\sum_{j \in \relInd(c)} q_j \cdot \rew_{T_j}\right) 
    \end{align*}
    and finally
    \begin{align*}
        \opt_{\sched \in \Scheds[\indexc]} \Expected^{\indexc, \sched}_{s_c}\left(\sum_{j \in \relInd(c)} q_j \cdot \rew_{T_j}\right) 
        &= 
        \opt_{\sched \in \Scheds[\indexc]} \Expected^{\indexc, \sched}_{s_c}(\indexc[\rew]) 
        ~,
    \end{align*}
    since 
    $\indexc[\rew] = \sum_{T \in \mathcal T_c} (\sum_{j \in \{\relInd(c)\mid T = T_j\}} q_j) \cdot \rew_{T} = \sum_{j \in \relInd(c)} q_j \cdot \rew_{T_j}$.    
    \hfill\qed
\end{proof}}{}

\begin{example}
    Following \Cref{ex:unfolding}, the (non-zero) rewards $\rew_c$ for $c = (s_1,\textcolor{gray}{\sched})$ are given in \textcolor{red}{red} next to the states in \Cref{fig:runningEx} (right).
\end{example}

\subsubsection{Step 3: Compute expected rewards.}
\label{step3}
The next step is to compute, for each individual scheduler-state combination $c$, the maximal and minimal rewards occurring in \Cref{thm:weightedReachProbsAsExpectedRew}.
Again, we rely on existing techniques from the literature~\cite{puterman1994markov,karmarkar1984new,quatmannVerificationMultiobjective2023}.\footnote{Note that we (must) rely on results supporting positive and negative reward, since we allow positive and negative coefficients.} 
\iftoggle{extended}{}{We refer to \cite{extendedVersion} for the proof.}

\begin{lemma}
    \label{thm:exRewPtimeMD}
    Let $c \in \comb$ and $\opt \in \{\max,\min\}$.
    The optimal expected reward $\opt_{\sched \in \Scheds[\indexc]} \Expected^{\indexc, \sched}_{s_c}(\indexc[\rew]) \in \Q$ is computable in time polynomial in the size of $\mdp_c$.
    Moreover, the optimum is attained by an MD scheduler $\sigma \in \SchedsMD[\mdp_c]$.
\end{lemma}

\iftoggle{extended}{\begin{proof}
    Existence of MD optimal schedulers is a well-known result, see, e.g.~\cite[Thm.~7.1.9]{puterman1994markov}.
    To compute the expected rewards in polynomial time, we rely on \emph{linear programming} (LP)~\cite{karmarkar1984new}.
    To express the $\max$-expected reward as the optimal solution of an LP we follow \cite[Ch.~4]{quatmannVerificationMultiobjective2023} and first cast the expected total reward objective to an \emph{expected reachability reward} objective~\cite[Def.~2.27]{quatmannVerificationMultiobjective2023} in a certain modified MDP $\mdp_c'$, the so-called \emph{MEC-quotient} of $\mdp_c$.
    Intuitively, $\mdp_c'$ arises from $\mdp_c$ by collapsing its MECs $E$ into individual states $s_E$, and adding a fresh absorbing state $s_\bot$ with an ingoing transition from every $s_E$ to model the option of staying within $E$ forever.
    Since every scheduler eventually remains inside some MEC, the state $s_{\bot}$ is \emph{attracting} in $\mdp_c'$, i.e., $\forall \sched \in \Scheds[\mdp_c'].~ \Pr^\sched(\Finally \{s_\bot\}) = 1$.
    Next, observe that $\rew_c$ (\Cref{def:rewForComb}) assigns zero reward to all states inside the ECs of $\mdp_c$ because reward is only collected upon reaching a target set \emph{for the first time}.
    We can thus naturally translate $\rew_c$ to a reward function $\rew_c'$ in the quotient $\mdp_c'$.
    
    The optimal expected reward in $\mdp_c$ is then equal to the optimal expected \emph{reachability reward} w.r.t.\ $\{\state_\bot\}$ in $\mdp_c'$, i.e., the expected reward collected until visiting $\state_\bot$ for the first time (see~\cite[Sec.~4.1.5]{quatmannVerificationMultiobjective2023}):
    \[
        \opt_{\sched \in \Scheds[\indexc]} \Expected^{\indexc, \sched}_{s_c}(\indexc[\rew])
        ~=~
        \opt_{\sched \in \Scheds[\mdp_c']} \Expected^{\indexc', \sched}_{s_c}(\indexc[\rew]' \Finally \{\state_\bot\})
    \]
    Since $\{\state_\bot\}$ is attracting, the \emph{Bellman-equations} associated with the above reachability reward objective in $\mdp_c'$ have a \emph{unique solution}~\cite[Thm.~4.8]{quatmannVerificationMultiobjective2023} and can thus be readily expressed as an LP of size linear in the size of $\mdp_c'$~\cite[Fig.~4.10]{quatmannVerificationMultiobjective2023}.
    
    The above approach applies to $\min$-expected rewards as well by noticing that in any MDP $\mdp$ with a state $\state$ and a reward function $\rew$ such that the optimal expected rewards are well-defined, $\min_{\sigma \in \Scheds} \Expected_s^\sched(\rew) = -\max_{\sigma \in \Scheds} \Expected_s^\sched(-\rew)$.
    \hfill\qed
\end{proof}}{}

\begin{example}
    \label{ex:expectedRewards}
    Reconsider the MDP in \Cref{fig:runningEx} (right).
    The maximal expected reward from initial state $(s_1,\emptyset)$ is $\tfrac 1 4$ and is attained by the MD strategy that always chooses $\alpha$ in $(s_1,\emptyset)$ and thus eventually reaches either $(t_2,\emptyset)$ or $(t_1,\emptyset)$ with probability $\tfrac 1 2$ each.
    In the latter case, the strategy then selects $\beta$ in $(s_1,\{T_1\})$ to reach $(s_2,\{T_1\})$ and remain there forever, not collecting any further reward.
    Overall, this strategy collects a total expected reward of $\tfrac 1 2 \cdot 1 + \tfrac 1 2 \cdot (-\tfrac 1 2) = \tfrac 1 4$.
    The minimal expected reward is easily seen to be $-\tfrac 1 2$.
\end{example}

\begin{remark}[Approximate vs exact]  
    \label{remark:approx_exact}
    In practice, \emph{exact} computation of the optimal expected reward via LP as suggested by \Cref{thm:exRewPtimeMD} (and its proof) may be significantly slower than approximation~\cite{DBLP:conf/tacas/HartmannsJQW23}.
    Fortunately, it is possible to amend our algorithm to \emph{approximate} expected reward computation.
    To retain soundness, it is crucial to employ a procedure such as \emph{Sound Value Iteration}~\cite{DBLP:conf/cav/QuatmannK18} that yields guaranteed \emph{under-} and \emph{over-approximations} of the true result.
    Appropriate handling of such approximations is detailed in \iftoggle{extended}{\Cref{alg:linear_general}}{\cite{extendedVersion}}.
    Note that approximation inherently leads to incompleteness, i.e., the algorithm may return \enquote{inconclusive} in some cases.
    Further, for each $c \in \comb$ with $|\relInd(c)|>1$ we can view $\opt_{\sched \in \Scheds[\mdp]} \sum_{j \in \relInd(c)} q_j \cdot \Pr_{s}^{\mdp,\sched}(\Finally T_j)$ as a \emph{weighted-sum optimization problem} and employ multi-objective model-checking techniques \cite{quatmannVerificationMultiobjective2023}. 
    (For $|\relInd(c)|=1$ this is a single-objective model-checking query.)
\end{remark}

\subsubsection{Step 4: Aggregate results and check relational property.}
\label{step4}
We now combine the optimal expected rewards $\opt_{\sched \in \Scheds[\indexc]} \Expected^{\indexc, \sched}_{s_c}(\indexc[\rew])$ for each state-scheduler combination $c \in \comb$ obtained via the previous two steps. 
We exemplify this for the comparison relation $\approx_\epsilon$, the other relations are handled similarly (see \iftoggle{extended}{\Cref{alg:linear_general}}{\cite{extendedVersion}}).

\begin{lemma}
    \label{thm:cruxApproxEqual}
    For each $c \in \comb$ and $\opt \in \{\max,\min\}$ let
    \[
        v_c^{\opt}
        ~=~
        \opt_{\sched \in \Scheds[\indexc]} \Expected^{\indexc, \sched}_{s_c}(\indexc[\rew])
        ~.
    \]
    Furthermore, let $v^{\opt} = \sum_{c \in \comb} v_c^{\opt}$.
    Then, assuming that the comparison operator $\comp$ in the property \eqref{eq:genRelReachProp} is $\approx_\epsilon$ for some $\epsilon \geq 0$:
    \[
        q_{m+1} \in [v^{\min} - \epsilon, v^{\max} + \epsilon] ~\iff~ \text{Property \eqref{eq:genRelReachProp} holds} ~.
    \]
\end{lemma}
\Cref{thm:cruxApproxEqual} relies on the fact that any value in the interval of achievable probabilities
$[v_c^{\min}, v_c^{\max}]$ can be achieved by constructing the \emph{convex combination} (e.g.,~\cite[p.\ 71]{quatmannVerificationMultiobjective2023}) of the minimizing and the maximizing scheduler.

\iftoggle{extended}{
    \begin{proof}[of \cref{thm:cruxApproxEqual}]
    We show \enquote{$\Rightarrow$}.
    It follows from \Cref{thm:whyComb,thm:weightedReachProbsAsExpectedRew} that the properties
    \begin{align*}
        & \exists \sched_1, \ldots, \sched_n \in \Scheds .~ \sum_{i=1}^{m} q_i \cdot \Pr^{\sched_{k_i}}_{\state_{i}}(\Finally T_i)
        ~\boldsymbol{\geq}~
        q_{m+1} \boldsymbol{-} \epsilon \qquad\text{and}\\
        & \exists \sched_1, \ldots, \sched_n \in \Scheds .~ \sum_{i=1}^{m} q_i \cdot \Pr^{\sched_{k_i}}_{\state_{i}}(\Finally T_i)
        ~\boldsymbol{\leq}~
        q_{m+1} \boldsymbol{+} \epsilon
    \end{align*}
    both hold in $\mdp$.    
    Let $\sched_1^{\geq},\ldots,\sched_n^{\geq}$ and $\sched_1^{\leq},\ldots,\sched_n^{\leq}$ be witnessing schedulers for the two properties.
    Further, let $\ub{v} \geq q_{m+1} - \epsilon$ and $\lb{v} \leq q_{m+1} + \epsilon$ be the weighted sums of probabilities induced by these schedulers.
    We distinguish two cases:
    
    First, if we have $\ub{v} \leq q_{m+1} + \epsilon$ \emph{or} $\lb{v} \geq q_{m+1} - \epsilon$, then $\sched_1^{\geq},\ldots,\sched_n^{\geq}$ or $\sched_1^{\leq},\ldots,\sched_n^{\leq}$ are already witnessing schedulers for the property and we are done.
    
    Otherwise, we have $\ub{v} > q_{m+1} + \epsilon$ \emph{and} $\lb{v} < q_{m+1} - \epsilon$.
    This means that $q_{m+1} \in (\lb v,\ub v)$ and thus there exists $\lambda \in (0,1)$ such that $q_{m+1} = \lambda \lb v + (1-\lambda) \ub v$.
    The crux is now to consider the schedulers $\sched_i^{\lambda} = [\sched^{\leq}_i \oplus_\lambda \sched^{\geq}_i]$, $i = 1,\ldots,n$ (for details, see \cite[p.~71]{quatmannVerificationMultiobjective2023}).
    These schedulers satisfy
    \begin{align*}
        \sum_{i=1}^{m} q_i \cdot \Pr^{\sched_{k_i}^{\lambda}}_{\state_{i}}(\Finally T_i) 
        ~=~&\lambda \cdot \sum_{i=1}^{m} q_i \cdot \Pr^{\sched_{k_i}^{\leq}}_{\state_{i}}(\Finally T_i) ~+~ (1-\lambda) \cdot  \sum_{i=1}^{m} q_i \cdot \Pr^{\sched_{k_i}^{\geq}}_{\state_{i}}(\Finally T_i) \\
        ~=~&\lambda \lb v + (1-\lambda) \ub v = q_{m+1}
        ~,
    \end{align*}
    hence they witness satisfaction of property \eqref{eq:genRelReachProp} with \emph{exact} equality ($\approx_0$).
    
    We show \enquote{$\Leftarrow$} by contraposition.
    Assume that $q_{m+1} < v^{\min} - \epsilon$ (the argument for the other case $q_{m+1} > v^{\min} + \epsilon$ is analogous).
    It follows that $q_{m+1} + \epsilon < \min_{\sched_1, \ldots, \sched_n} \sum_{i=1}^{m} q_i \cdot \Pr^{\sched_{k_i}}_{\state_{i}}(\Finally T_i)$, i.e., all schedulers give a value at least $\epsilon$ apart from $q_{m+1}$ and hence property \eqref{eq:genRelReachProp} does not hold.
    \hfill\qed
\end{proof}
}{
\begin{proof}[of \Cref{thm:cruxApproxEqual}; sketch]
    The interesting direction is ``$\Rightarrow$'':
    By \Cref{thm:whyComb,thm:weightedReachProbsAsExpectedRew} there exist schedulers $\sched_1^{\geq},\ldots,\sched_n^{\geq}$ and $\sched_1^{\leq},\ldots,\sched_n^{\leq}$ for $\mdp$ such that   
    \begin{align*}
        \ub{v} := \sum_{i=1}^{m} q_i \cdot \Pr^{\sched^{\geq}_{k_i}}_{\state_{i}}(\Finally T_i)
        ~\geq~
        q_{m+1} - \epsilon~, \quad
        \lb{v} := \sum_{i=1}^{m} q_i \cdot \Pr^{\sched^{\leq}_{k_i}}_{\state_{i}}(\Finally T_i)
        ~\leq~
        q_{m+1} + \epsilon ~ .
    \end{align*}
    If $\ub{v} \leq q_{m+1} + \epsilon$ \emph{or} $\lb{v} \geq q_{m+1} - \epsilon$, then $\sched_1^{\geq},\ldots,\sched_n^{\geq}$ or $\sched_1^{\leq},\ldots,\sched_n^{\leq}$ are already witnessing schedulers and there is nothing else to show.
    Otherwise, $\ub{v} > q_{m+1} + \epsilon$ \emph{and} $\lb{v} < q_{m+1} - \epsilon$.
    Thus there exists $\lambda \in (0,1)$ such that $q_{m+1} = \lambda \lb v + (1-\lambda) \ub v$.
    The schedulers $\sched_i^{\lambda} = [\sched^{\leq}_i \oplus_\lambda \sched^{\geq}_i]$, $i = 1,\ldots,n$ \cite[p.~71]{quatmannVerificationMultiobjective2023} then witness satisfaction of property \eqref{eq:genRelReachProp} with \emph{exact} equality ($\approx_0$).
    See \cite{extendedVersion} for more details.
    \hfill\qed
\end{proof}
}

\begin{example}
    We wrap up our running example by proving that the property from \Cref{ex:runningExIntro} indeed holds in the MDP in \Cref{fig:runningEx} (left).
    For combination $c_1 = (s_1,\textcolor{gray}{\sched})$  we have already established in \Cref{ex:expectedRewards} that $v_{c_1}^{\max} = \tfrac 1 4$ and $v_{c_1}^{\min} = -\tfrac 1 2$.
    For the other combination $c_2$ one easily finds $v_{c_2}^{\max} = 0$ and $v_{c_2}^{\min} = -\tfrac 1 2$. Summing up these values yields $v^{\min} = -1$ and $v^{\max} = \tfrac 1 4$.
    Since $0 \in [v^{\min}, v^{\max}]$, the property is satisfiable, even with exact equality $\approx_0$.
\end{example}

We remark that for $\comp~ \in \{\geq, >\}$, it actually suffices to compute \emph{only} $v^{\max}$ rather than both $v^{\max}$ and $v^{\min}$ as in \Cref{thm:cruxApproxEqual}.

\subsubsection{Overall algorithm.}

\iftoggle{extended}{
    \begin{algorithm}[t]
    \caption{Approximate expected reward (black box)~\cite{DBLP:conf/cav/QuatmannK18},\cite[Alg.~4.6]{quatmannVerificationMultiobjective2023}}
    \Input{
        MDP $\mdp=\mdptup$,
        reward function $\rew \colon \states \to \Q$,
        initial state $s \in \states$,
        \emph{attracting} set $\attractor \subseteq \states$ (i.e., $\forall \sigma \in \Scheds[\mdp] .~ \Pr_s(\Finally \attractor) = 1$),
        absolute tolerance $\tau \geq 0$ (use $\tau = 0$ for exact computation) 
    }
    \Output{
        $\lb{v}, \ub{v} \in \Q$ such that $\lb{v} \leq \max_{\sched \in \Scheds[\mdp]}\Expected_s^{\sched}(\rew^{\Finally \attractor}) \leq \ub{v}$ and $\ub{v} - \lb{v} \leq \tau$
    }
    \label{alg:black-box}
\end{algorithm}

\begin{algorithm}
    \caption{Efficient Solution of \RelReach} 
    \label{alg:linear_general}
    \Input{
        MDP $\mdp = \mdptup$, \RelReach property \medskip \newline 
        {\centering
        $\exists \sched_1, \ldots, \sched_n \in \Scheds .~ \sum_{i=1}^{m} q_i \cdot \Pr^{\sched_{k_i}}_{\state_{i}}(\Finally T_i) \comp q_{m+1}$}
        ~,
        \medskip \newline 
        and tolerance $\tau \geq 0$ ($\tau = 0$ yields exact computation)
    }
    \Output{Whether the property is true in $\mdp$, or \enquote{inconclusive} (the latter can only happen if $\tau > 0$)}
    \tcp{Step 1: Loop over all state-scheduler combinations:}
    \For{$c = (s,\sched) \in \comb = \{ (s_i, \sched_{k_i}) \mid i=1, \ldots, m\}$}{
        \tcp{Step 2: Unfold and define reward structures:}
        $\indexc \gets$ goal unfolding of $\mdp$ w.r.t.\ target sets for $c$ \tcp*{See \Cref{def:goalUnfolding}}
        $s_{c} \gets $ the state $(s,\emptyset)$ in $\indexc$ \tcp*{Just for readability}
        $\indexc[\rew] \gets$ reward structure on $\indexc$ for $c$ \tcp*{See \Cref{def:rewForComb}}
        \tcp{Step 3: Compute (or approximate) expected rewards:}
        $\mdp_c', \state_\bot \gets$ MEC-quotient of $\mdp_c$ and its absorbing state \;
        $\indexc[\lb{v}]^{\max}, \indexc[\ub{v}]^{\max} \gets \mathsf{AprxExRew}(\indexc[\mdp'],\indexc[\rew],s_{c},\{\state_\bot\}, \tau)$  \tcp*{using the black box} 
        \If{$\comp~\in \{ \approx_\epsilon, \not\approx_\epsilon \}$}{
            \tcp{Compute $\min$-expected reward by flipping signs}
            $\indexc[\ub{v}]^{\min}, \indexc[\lb{v}]^{\min} \gets (-1) \cdot \mathsf{AprxExRew}(\indexc[\mdp'],-\indexc[\rew],s_{c},\{\state_\bot\}, \tau)$ \;
        }
    }
    \tcp{Step 4: Aggregate results from individual state-scheduler combinations and check appropriate conditions (depending on $\comp$):}
    $\lb{v}^{\max} \gets \sum_{c \in \comb} \indexc[\lb{v}]^{\max} \,;~~ \ub{v}^{\max}  \gets \sum_{c \in \comb} \indexc[\ub{v}]^{\max}$ \;
    \If{$\comp~\in \{ >, \geq \}$}{
        \lIf{\hspace{7.2mm}$\lb{v}^{\max} \comp q_{m+1}$}{
            \Return{true} \label{line:linear_general_return_first}
        }
        \lElseIf{$\ub{v}^{\max} \not\comp q_{m+1}$}{
            \Return{false}
        }
        \lElse{\Return{\enquote{inconclusive}}}
    }
    \ElseIf{$\comp~\in \{ \approx_\epsilon, \not\approx_\epsilon \}$}{
        $\lb{v}^{\min} \gets \sum_{c \in \comb} \indexc[\lb{v}]^{\min} \,;~~ \ub{v}^{\min}  \gets \sum_{c \in \comb} \indexc[\ub{v}]^{\min}$ \;
        \If{$\comp$ is $\approx_\epsilon$}{
            \lIf{\hspace{7.2mm}$q_{m+1} \in [\ub{v}^{\min} - \epsilon, \lb{v}^{\max} + \epsilon]$}{
                \Return{true}
            }
            \lElseIf{$q_{m+1} \notin [\lb{v}^{\min} - \epsilon, \ub{v}^{\max} + \epsilon]$}{
                \Return{false}
            }
            \lElse{\Return{\enquote{inconclusive}}}
        }
        \ElseIf{$\comp$ is $\not\approx_\epsilon$}{
            \lIf{\hspace{7.2mm}$q_{m+1} \notin [\lb{v}^{\max} - \epsilon, \ub{v}^{\min} + \epsilon]$}{
                \Return{true}
            }
            \lElseIf{$q_{m+1} \in [\ub{v}^{\max} - \epsilon, \lb{v}^{\min} + \epsilon]$}{
                \Return{false}
                 \label{line:linear_general_return_last}
            }
            \lElse{\Return{\enquote{inconclusive}}}
        }
    }
\end{algorithm}
}{
    \begin{algorithm}[t]
        \caption{Solving \RelReach with $\approx_\epsilon$} 
        \label{alg:linear_general_simplified}
        \Input{
            MDP $\mdp = \mdptup$ and a \RelReach property\\\medskip
            \centering
            \enquote{$\exists \sched_1, \ldots, \sched_n \in \Scheds .~ \sum_{i=1}^{m} q_i \cdot \Pr^{\sched_{k_i}}_{\state_{i}}(\Finally T_i) \approx_\epsilon q_{m+1}$}
            \medskip
        }
        \Output{Whether the property is true in $\mdp$}
        \tcp{Step 1: Loop over all state-scheduler combinations:}
        \For{$c = (s,\sched) \in \comb = \{ (s_i, \sched_{k_i}) \mid i=1, \ldots, m\}$}{
            \tcp{Step 2: Unfold and define reward structures:}
            $\indexc \gets$ goal unfolding of $\mdp$ w.r.t.\ target sets for $c$ \tcp*{See \Cref{def:goalUnfolding}}
            $\indexc[\rew] \gets$ reward structure on $\indexc$ for $c$ \tcp*{See \Cref{def:rewForComb}}
            \tcp{Step 3: Compute (or approximate) expected rewards:}
            $\indexc[{v}]^{\max} \gets  {\displaystyle\max_{\sched \in \Scheds[\indexc]}} \Expected^{\indexc, \sched}_{s_c}(\indexc[\rew]) ~;~
            \indexc[{v}]^{\min} \gets  {\displaystyle\min_{\sched \in \Scheds[\indexc]}} \Expected^{\indexc, \sched}_{s_c}(\indexc[\rew])$
            \label{line:linear_general_simplified_step3}
        }
        \tcp{Step 4: Aggregate results from state-scheduler combinations and check:}
        $v^{\max} \gets \sum_{c \in \comb} \indexc[{v}]^{\max} ~;~ {v}^{\min}  \gets \sum_{c \in \comb} \indexc[{v}]^{\min}$ \;
        \label{line:linear_general_simplified_step4}
        \Return{$q_{m+1} \in [v^{\min} - \epsilon, v^{\max} + \epsilon]$}
        \label{line:linear_general_simplified_return}
    \end{algorithm}
}

\iftoggle{extended}{
    The algorithm resulting from Steps 1--4 is stated explicitly as \cref{alg:linear_general}, using a black box for the approximate computation of expected rewards~\cite{DBLP:conf/cav/QuatmannK18},\cite[Alg.~4.6]{quatmannVerificationMultiobjective2023} (\cref{alg:black-box}).
}{
    The algorithm resulting from Steps 1--4 is stated explicitly as \cref{alg:linear_general_simplified} for $\approx_\epsilon$ and in full generality in \cite{extendedVersion}.
}

\begin{theorem}[Correctness and time complexity]
    \label{thm:correctness}
    \iftoggle{extended}{%
        \cref{alg:linear_general}%
    }{%
        \cref{alg:linear_general_simplified}%
    }
    adheres to its input-output specification.
    It can be implemented with worst-case running time of $\mathcal{O}(m \cdot \textit{poly}(2^m \cdot |\mdp|))$, where $|\mdp|$ is the size of (an explicit encoding of) $\mdp$ and $m$ is the number of probability operators in the property.
\end{theorem}
\begin{proof}
    \Cref{thm:cruxApproxEqual} establishes correctness.
    Regarding time complexity, notice that 
    for each $c \in \comb$, the size of the goal unfolding $\indexc$ is bounded by $2^m \cdot |\mdp|$ (Step 2).
    Exact computation of expected rewards is possible in time polynomial in the size of the MDP $\indexc$ (Step 3).
    Steps 2 and 3 have to be executed for at most $|\comb|\leq m$ state-scheduler combinations.
    \hfill\qed
\end{proof}

\section{Complexity of Relational Reachability}
\label{sec:complexity}

In this section, we analyze the computational complexity of the \RelReach problem over general and over memoryless deterministic schedulers, respectively.
We also identify restricted variants of \RelReach that are decidable in polynomial time.

\subsection{General Schedulers}

The runtime analysis from \cref{thm:correctness} yields an \EXPTIME upper bound for the complexity of the \RelReach problem.

\begin{theorem}
    \label{th:general_EXPTIME}
    Problem \RelReach is \PSPACE-hard and decidable in \EXPTIME.
\end{theorem}
\begin{proof}
    For \PSPACE-hardness observe that \emph{simultaneous almost-sure reachability} of $m$ target sets $T_1,\ldots,T_m$ (which is known to be \PSPACE-complete~\cite[Theorem 2]{DBLP:conf/cav/RandourRS15}) is expressible as the \RelReach property \enquote{$\exists \sigma . \Pr_{s}^\sigma(\Finally T_1) + \ldots \Pr_{s}^\sigma(\Finally T_m) \geq m$}.

    Membership in \EXPTIME follows from \cref{thm:correctness}.
    \hfill\qed
\end{proof}

Tighter bounds are ongoing work.
Taking a closer look,  
we however observe that \RelReach can be solved in \PTIME if the number of probability operators $m$ is fixed. 
Hence, \RelReach is \emph{fixed-parameter tractable} \cite{groheDescriptiveParameterized1999} with parameter $m$. 
The following theorem generalizes this observation and further states that the exponential blow-up of the goal unfolding can be avoided if all target states are absorbing, or if each probability operator is evaluated under a different scheduler (i.e., if $n=m$).
\iftoggle{extended}{}{We refer to \cite{extendedVersion} for the proof.}

\begin{theorem}
    \label{th:general_PTIME}
    The following special cases of \RelReach are in \PTIME:
    \begin{enumerate}[label=(\alph*)]
        \item \label{item:fixed-param}
        The number of \emph{different} target sets $|\{ T_1,\ldots,T_m\}|$ is \emph{at most a constant}.

        \item \label{item:absorbing}
        The target sets $T_1,\ldots,T_m$ are all \emph{absorbing}.

        \item \label{item:independent}
        $n=m$, i.e., each probability operator in the property has \emph{its own quantifier}.
    \end{enumerate}
\end{theorem}

\iftoggle{extended}{\begin{proof}
    \ref{item:fixed-param}: $m$ is the only parameter occurring exponentially in the runtime complexity of \cref{alg:linear_general} with exact reward computation.
    
    \ref{item:absorbing}: If all target sets are absorbing, then for each $c \in \comb$, the number of reachable states in $\indexc$ is $|(\states \times \emptyset) \cup \bigcup_i T_i \times \{T_i\}| = |\states| + \sum_{i=1}^{m} |T_i| \leq (m+1)\cdot|\states|$. 
    Hence, the size of the unfolding is linear in the size of the original MDP and $m$.
    
    \ref{item:independent}: If each probability operator is evaluated under a different scheduler, then $|\comb| = n= m$ and for each $c \in \comb$ we have $|\relInd(c)| = 1$. Hence, for each $c \in \comb$, we have $|\indexc[\states]| = |\states \times 2^{\relInd(c)}| = |\states| \cdot 2$ and thus the size of the unfolding with respect to $\relInd(c)$ is linear in the size of the original MDP.
    \hfill\qed
\end{proof}}{}

\subsection{Memoryless Deterministic Schedulers}

We now consider \RelReachMD, the \RelReach problem over MD schedulers. \RelReachMD is in \NP because we can non-deterministically guess schedulers and verify whether they are witnesses in polynomial time by computing the (exact) reachability probabilities in the induced DTMC~\cite[Ch.~10]{bk08-book}.
Further, \RelReach is \emph{strongly} \NP-hard\footnote{A problem is strongly \NP-hard if it is \NP-hard even if all numerical quantities (here: rational transition probabilities) in a given input instance are encoded in unary.
} \cite{gareyStrongNPCompleteness1978}
over MD schedulers already for simple variants with equality.

\begin{theorem}
    \label{th:MD_NP-complete}
    \label{th:MD_NP}
    \RelReachMD is strongly \NP-complete.
    Strong \NP-hardness already holds for the following special cases:
    For a given MDP $\mdp$, 
    initial states $\state_1,\state_2 \in \states$,
    target sets $T_1, T_2 \subseteq \states$,
    decide if
    \begin{enumerate}[label=\textnormal{(\alph*)},itemsep=2pt]
        \item 
        \label{item:1sched1state}
        \label{item:1sched2state}
        $
            \exists \sched \in \SchedsMD .\ 
            \Pr^{\sched}_{\state_1}(\Finally T_1) - \Pr^{\sched}_{\state_2}(\Finally T_2) = 0 \ .
        $
        
        \item 
        \label{item:2sched2state}
        $
            \exists \sched_1, \sched_2 \in \SchedsMD .\ 
            \Pr^{\sched_1}_{\state_1}(\Finally T_1) - \Pr^{\sched_2}_{\state_2}(\Finally T_2) = 0 \ .
        $
    \end{enumerate}
    Strong \NP-hardness of \ref{item:1sched1state} and \ref{item:2sched2state} holds irrespective of whether 
    $\state_1 = \state_2$
    and whether $T_1$ and/or $T_2$ are absorbing.
    Moreover, \ref{item:1sched1state} and \ref{item:2sched2state} with relation $\approx_\epsilon$, $\epsilon > 0$, are \NP-hard.
\end{theorem}

\iftoggle{extended}{
    We show \NP-hardness by giving a polynomial transformation from the Hamiltonian path problem, 
which is known to be strongly \NP-hard~\cite{gareyStrongNPCompleteness1978}, inspired by~\cite{footePolynomialTime}.
For exact equality, we establish strong \NP-hardness by showing that our transformation is pseudo-polynomial.

\begin{definition}[Hamiltonian Path Problem]
    Given a directed graph $G = (V, E)$ 
    (where $V$ is a set of vertices and $E \subseteq V \times V$ a set of edges) 
    and some initial vertex $v \in V$,
    decide whether there exists a path from $v$ in $G$ that visits each vertex exactly once.
\end{definition}

\begin{figure}[t]
    \centering
    \begin{subfigure}[b]{0.3\textwidth}
        \centering
        \begin{tikzpicture}
            \node[state] (v0) {$v_\init$};
            \node[state, below= of v0] (v1) {$v_1$};
            \node[state, right= of v0] (v2) {$v_2$};
            \node[state, right= of v1, xshift=0.1cm] (v3) {$v_3$};

    
            \path[-latex', draw]
                (v0) edge (v1)
                (v1) edge (v2)
                (v1) edge (v3)
                (v2) edge (v0)
                (v2) edge (v3);
        \end{tikzpicture}%
    \caption{Graph $G = (V, E)$}
    \label{fig:ill_ham-path_reduction_graph}
    \end{subfigure}%
    ~
    \begin{subfigure}[t]{0.7\textwidth}
        \centering
    \begin{tikzpicture}
        \node[state] (v0) {$v_\init$};
        \node[state, below= of v0] (v1) {$v_1$};
        \node[state, right= of v0] (v2) {$v_2$};
        \node[state, right= of v1, xshift=0.1cm] (v3) {$v_3$};
        
        \node[dist] at ($(v0)!0.4!(v1)$) (d01) {};
        \node[dist] at ($(v0)!0.6!(v2)$) (d02) {};
        \node[dist] at ($(v1)!0.4!(v2)$) (d12) {};
        \node[dist] at ($(v1)!0.4!(v3)$) (d13) {};
        \node[dist] at ($(v2)!0.4!(v3)$) (d23) {};
        
        \node[] at ($(v2)!0.5!(v3)$) (sinkhelp) {};
        \node[state, right= of sinkhelp, color=gray] (sink) {$s_\bot$};
        
        \node[] at ($(v0)!0.5!(v1)$) (sahelp) {};
        \node[state, left= of sahelp, color=blue] (sa) {$s_a$}; 

        \node[right of = v2, xshift=3cm] (s1help) {};
        \node[state, yshift=1.5cm] at (s1help) (s1) {$s_1$};  
        \node[dist, below of=s1, yshift=0.5cm] (ds1) {};
        \node[below of=ds1] (ds2) {\vdots};
        \node[state, below of=ds2] (svm1h) {\phantom{$s_{|V|}$}};
        \node at (svm1h) (svm1) {$s_{\scriptscriptstyle |V|-1 }$};
        \node[dist, below of=svm1h, yshift=0.3cm] (dsvm1) {};

        \node[state, below= of svm1h, yshift=0.5cm, color=purple] (sb) {$s_b$}; 

        \node[state, yshift=1.5cm] at ($(v0)!0.5!(s1help)$) (sinit) {$\statei$};
        \node[dist, below= of sinit, yshift=0.75cm] (dinit) {};
        
        \path[-latex', draw]
            (v0) edge (v1)
            (v1) edge (v2)
            (v1) edge (v3)
            (v2) edge (v0)
            (v2) edge (v3);

        \path[-, draw]
            (sinit) edge (dinit)
            (s1) edge (ds1) 
            (svm1h) edge (dsvm1)
            ;
        \path[-latex', draw]
            (dinit) edge[color=blue, bend right=10] node[above] {} (v0)
            (dinit) edge[color=purple, bend right=10] node[above] {} (s1)
            (ds1) edge[color=gray] node[above] {} (sink)
            (ds1) edge[color=purple] node[right] {} (ds2)
            (dsvm1) edge[color=gray] node[above] {} (sink)
            (dsvm1) edge[color=purple] node[right] {} (sb); 

        \path[-latex', draw, color=gray]
            (d01) edge[bend left=20] node[pos=0.9, above] {} (sink)
            (d02) edge[bend left=80] node[pos=0.8, right] {} (sink)
            (d12) edge[bend right=20] node[pos=0.8, below] {} (sink)
            (d13) edge[bend right=80] node[pos=0.8, right] {} (sink)
            (d23) edge[bend right=10] node[pos=0.3, below] {} (sink);
            
        \path[-latex', draw, color=blue]
            (v0) edge[bend right] (sa)
            (v1) edge[bend left] (sa)
            (v2) edge[bend right=50] (sa)
            (v3) edge[bend left=50] (sa);

        \path[-latex', draw]
            (sink) edge[in=300, out=330, loop, color=gray] (sink)
            (sa) edge[loop below, color=blue] (sa)
            (sb) edge[loop right, color=purple] (sb);
    \end{tikzpicture}
    \caption{MDP constructed from $G$ and initial vertex $v_0$. All distributions are uniform. }
    \label{fig:ill_ham-path_reduction_mdp}
    \end{subfigure}
    \caption{Illustration of the MDP construction for the reduction from the Hamiltonian path problem. 
    }
    \label{fig:ill_ham-path_reduction}
\end{figure}
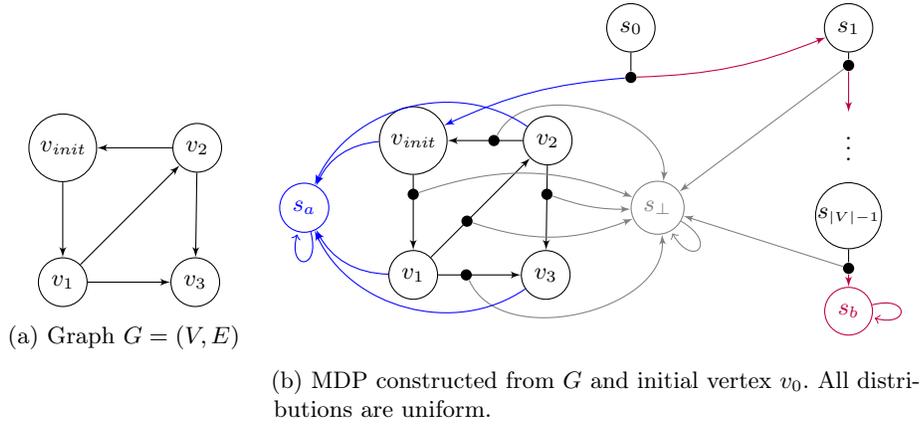

\begin{proof}[of \cref{th:MD_NP-complete}]
    Membership: Given some \RelReach property, MDP $\mdp$ and a memoryless deterministic scheduler $\sched \in \Scheds$, we can verify whether $\sched$ is a witness for the property 
    by computing the (exact) reachability probabilities in the induced DTMC, see e.g.,~\cite[Ch.~10]{bk08-book}.
    This is possible in time polynomial in the size of the state space~\cite{10.1007/3-540-60692-0_70}.
    
    \NP-hardness:
    We first give the reduction for \ref{item:1sched1state} with $\state_1 = \state_2$
    and explain how to adjust the construction for the other cases afterwards.
    For a given instance of the Hamiltonian path problem $G = (V, E)$ and $v_{\init} \in V$, 
    we construct the MDP $\mdp = \mdptup$ with
    \begin{itemize}
        \item $\states = V \cup \{ \statei, s_\bot, s_a, s_b \} \cup \{ \state_i \mid i = 1, \ldots, |V|-1\}$
        
        \item $\Act = E \cup \{ \tau \}$
        
        \item 
        $\Trans(\statei, \tau, v_\init) = \frac{1}{2}$ and $\Trans(s_0, \tau, s_1) = \frac{1}{2}$ \\
        $\Trans(s_{i}, \tau, s_{i+1}) = \frac{1}{2}$ and $\Trans(s_{i}, \tau, s_\bot) = \frac{1}{2}$ for $i=1,\ldots, |V|-2$\\
        $\Trans(s_{|V|-1}, \tau, s_b) = \frac{1}{2}$ and $\Trans(s_{|V|-1}, \tau, s_\bot) = \frac{1}{2}$ \\
        $\Trans(v, (v,v'), v') = \frac{1}{2}$ and $\Trans(v, (v,v'), s_\bot) = \frac{1}{2}$ for all $v, v' \in V, (v,v') \in E$ \\
        $\Trans(v, \tau, s_a) = 1$ for all $v \in V$ \\
        $\Trans(s_\bot, \tau, s_\bot) = 1$, $\Trans(s_a, \tau, s_a) = 1$ and $\Trans(s_b, \tau, s_b) = 1$ \\
        $\Trans(s, \alpha, s') = 0$ otherwise
    \end{itemize}
    \cref{fig:ill_ham-path_reduction} illustrates the construction: 
    \cref{fig:ill_ham-path_reduction_mdp} shows the MDP constructed from the graph $G$ with initial vertex $v_0$ depicted in \cref{fig:ill_ham-path_reduction_graph}.

    We observe that for all schedulers $\sched \in \SchedsMD$, it holds that $\Pr^\sched_{\statei}(\Finally \{s_b\}) = \frac{1}{2^{|V|}}$.
    Further, for all schedulers $\sched \in \SchedsMD$, we have
    \begin{align*}
        {\Pr}^{\sched}_{\statei}(\Finally \{s_a\}) 
        &= \frac{1}{2} \cdot {\Pr}^{\sched}_{v_{\init}}(\Finally \{s_a\}) 
    \end{align*}

    \paragraph{Claim:} \emph{For any $\epsilon < \frac{1}{2^{|V|+1}}$, it holds that there exists a Hamiltonian path from $v_\init$ in $G$ iff there exists some $\sched \in \SchedsMD$ such that $\Pr^\sched_{\statei}(\Finally \{s_a\}) - \Pr^\sched_{\statei}(\Finally \{s_b\}) \approx_\epsilon 0$.}

    ``$\Rightarrow$'': 
    Assume there exists a Hamiltonian path from $v_\init$ in $G$.
    Then, we construct $\sched$ by following this path in $\mdp$ and transitioning to $s_\bot$ from the last vertex. 
    By construction, it holds that ${\Pr}^{\sched}_{\statei}(\Finally \{s_a\}) = \frac{1}{2} \cdot \Pr^{\sched}_{v_{\init}}(\Finally \{s_a\}) = \frac{1}{2^{|V|}}$.
    Since $\Pr^{\sched}_{\statei}(\Finally \{s_b\}) = \frac{1}{2^{|V|}}$ for all schedulers $\sched \in \SchedsMD$, this implies $\Pr^{\sched}_{\statei}(\Finally \{s_a\}) = \Pr^{\sched}_{\statei}(\Finally \{s_b\})$ and hence in particular also $\Pr^\sched_{\statei}(\Finally a) - \Pr^\sched_{\statei}(\Finally b)| \approx_\epsilon 0$.

    ``$\Leftarrow$'': 
    Assume there exists $\sched \in \SchedsMD$ such that $|\Pr^\sched_{\statei}(\Finally \{s_a\}) - \Pr^\sched_{\statei}(\Finally \{s_b\})| \leq \epsilon$.
    We first show that $\Pr^\sched_{\statei}(\Finally \{s_a\}) > 0$: Assume $s_a$ is not reachable from $\statei$. Then $\frac{1}{2^{|V|+1}} > \epsilon \geq |\Pr^\sched_{\statei}(\Finally \{s_a\}) - \Pr^\sched_{\statei}(\Finally \{s_b\})| = |\Pr^\sched_{\statei}(\Finally \{s_b\})| = \frac{1}{2^{|V|}}$, which is a contradiction.
    So $s_a$ is reachable from $\statei$, and hence also from $v_\init$ and the finite path from $v_\init$ to $s_a$ cannot contain a loop since since $\sched$ is memoryless deterministic.
    Let $n$ be the number of $V$-states on the path from $v_\init$ to $s_a$ under $\sched$ (not counting $v_\init$ itself).
    Then $\Pr^{\sched}_{\statei}(\Finally \{s_a\}) = \frac{1}{2^{n+1}}$.
    Since $n \leq |V|-1$, this implies that $\Pr^{\sched}_{\statei}(\Finally \{s_a\}) \geq \frac{1}{2^{|V|}}$.
    Further, 
    \begin{align*}
        |\Pr^{\sched}_{\statei}(\Finally \{s_a\}) - \Pr^{\sched}_{\statei}(\Finally \{s_b\})| \leq \epsilon 
        \iff
        \Pr^{\sched}_{\statei}(\Finally \{s_a\}) - \frac{1}{2^{|V|}} \leq \epsilon ~ .
    \end{align*}
    Since $\epsilon < \frac{1}{2^{|V|+1}}$, this implies 
    \[ 
        \Pr^{\sched}_{\statei}(\Finally \{s_a\}) < \frac{1}{2^{|V|}} + \frac{1}{2^{|V|+1}} < \frac{1}{2^{|V|-1}} ~ .
    \]
    From $\frac{1}{2^{|V|-1}} > \Pr^{\sched}_{\statei}(\Finally \{s_a\}) = \frac{1}{2^{n+1}} \geq \frac{1}{2^{|V|}}$ we can conclude that $n = |V|-1$.
    Hence, the path from $v_\init$ to $s_a$ visits $|V|-1$ states corresponding to vertices in $G$, which means it must visit each vertex of $G$ exactly once and thus corresponds to a Hamiltonian path from $v_\init$ in $G$.

    \paragraph{Claim:} \emph{The above construction defines a pseudo-polynomial time transformation to \ref{item:1sched1state}, and a polynomial-time transformation to \ref{item:1sched1state} with approximate equality with $\epsilon>0$.}
    
    The constructed MDP has $2 \cdot |V| + 3$ states and $|E|+1$ actions. Each state has at most two successors.
    All transition probabilities in the MDP are either $0$, $0.5$, or $1$. Hence, the magnitude of the largest number occurring in the constructed MDP is a constant and thus polynomial in the size of the original graph.

    For approximate equality, we must choose some $0 < \epsilon < \frac{1}{2^{|V|+1}}$ for the constructed property. 
    The magnitude of the largest number occurring in the constructed \RelReach instance thus depends exponentially on the size of the original Hamiltonian path problem instance.
    Hence, the transformation is not pseudo-polynomial.
    The transformation is, however, still polynomial since clearly $\epsilon <1$ and hence all numbers are polynomially bounded.

    For exact equality, however, $\epsilon=0$ and hence the magnitude of the largest number occurring in the constructed \RelReach instance is polynomial in the size of the Hamiltonian path problem instance.
    Hence, the transformation is pseudo-polynomial.

    \paragraph{Handling the other cases.}
    \begin{itemize}
        \item \ref{item:2sched2state}, $s_1 = s_2$:
        We construct the MDP as above, and the following property:
        \enquote{$\exists \sched_1, \sched_2 \in \SchedsMD .\ \Pr^{\sched_1}_{s_0}(\Finally \{s_a\}) - \Pr^{\sched_2}_{s_0}(\Finally \{s_b\}) \approx_\epsilon 0$}.
        The proof works completely analogously since for all schedulers $\sched_1 \in \SchedsMD$, it holds that $\Pr^{\sched_1}_{\statei}(\Finally \{s_b\}) = \frac{1}{2^{|V|}}$
        and for all schedulers $\sched_2 \in \SchedsMD$, we have
       ${\Pr}^{\sched_2}_{\statei}(\Finally \{s_a\}) = \frac{1}{2} \cdot {\Pr}^{\sched_2}_{v_{\init}}(\Finally \{s_a\})$.        
        
        \item \ref{item:1sched2state}, $s_1 \neq s_2$:
        The MDP construction works analogously to the construction above, the only difference being that we do not introduce a fresh initial state. 
        We construct the property
        \enquote{$\exists \sched \in \SchedsMD .\ \Pr^\sched_{v_0}(\Finally \{s_a\}) - \Pr^\sched_{s_1}(\Finally \{s_b\}) \approx_\epsilon 0$}.
        Observe that for all schedulers $\sched \in \SchedsMD$, it holds that $\Pr^{\sched}_{s_1}(\Finally \{s_b\}) = \frac{1}{2^{|V|-1}}$.
        The proof works analogously.

        \item \ref{item:2sched2state}, $s_1 \neq s_2$:
        Again, we construct the MDP as above but without the fresh initial state.
        We construct the property \enquote{$\exists \sched_1, \sched_2 \in \SchedsMD .\ \Pr^{\sched_1}_{v_0}(\Finally \{s_a\}) - \Pr^{\sched_2}_{s_1}(\Finally \{s_b\}) \approx_\epsilon 0$}.
        The proof works analogously.
        \hfill \qed
    \end{itemize}
    
\end{proof}
}{
    \begin{proof}[Sketch]
        We show strong \NP-hardness by giving a pseudo-polynomial transformation from the \emph{Hamiltonian path} problem, which is known to be strongly \NP-hard \cite{gareyStrongNPCompleteness1978}, inspired by \cite{footePolynomialTime}.
        \NP-hardness of the cases for approximate equality follows by an analogous transformation, but for $\epsilon$ the transformation is only polynomial, not pseudo-polynomial, hence establishing \NP-hardness but not strong \NP-hardness.
        We refer to \cite{extendedVersion} for the full proof.
        \hfill\qed
    \end{proof}
}

Note that the hardness of the problem does not rely on whether all probability operators are evaluated under the same scheduler and from the same initial state. 

We observe that Steps 1--4 detailed in \cref{sec:algo} construct memoryful randomized witness schedulers in case of approximate equality, if they exist.
Memory and/or randomization are necessary for constructing a scheduler that exactly achieves some specified reachability probability, in general\iftoggle{extended}{}{ (see \cite{extendedVersion} for details)}.

\begin{theorem}
    Memory and randomization are necessary for \RelReach with approximate equality.
    \label{th:general_necessary}
\end{theorem}

\iftoggle{extended}{
    \begin{figure}[t]
    \centering
    \begin{subfigure}[b]{0.49\textwidth}
        \centering
    \begin{tikzpicture}
        \node[state] (sinita) {$s$};  
        
        \node[state, below left= of sinita] (t) {\phantom{$s_{2}$}}; 
        \node at (t) (th) {$t$};
        \node[state, below right= of sinita] (sx) {$s_\bot$}; 

        \path[-latex', draw]
            (sinita) edge node[left] {$\alpha$} (t)
            (sinita) edge node[right] {$\beta$} (sx);

        \path[-latex', draw]
            (t) edge[loop left] (t)
            (sx) edge[loop right] (sx);
    \end{tikzpicture}
    \caption{$\exists \sched .\ \Pr^\sched_s(\Finally \{t\}) \approx_\epsilon 0.5$ 
    }
    \label{fig:randomization-necessary}
    \end{subfigure}%
    \begin{subfigure}[b]{0.49\textwidth}
        \centering
    \begin{tikzpicture}
        \node[state] (s) {$s$};
        
        \node[state, below left= of s] (t1) {$t_1$}; 
        \node[state, below right= of s] (t2) {$t_2$}; 

       \path[-latex', draw]
            (s) edge[bend left] node[right] {$\alpha$} (t1)
            (s) edge node[right] {$\beta$} (t2);

        \path[-latex', draw]
            (t1) edge[bend left] node[right] {$\gamma$} (s)
            (sx) edge[loop right] (sx);
    \end{tikzpicture}
    \caption{$\exists \sched .\ \Pr^{\sched}_{s}(\Finally \{t_1\}) = \Pr^{\sched}_{s}(\Finally \{t_2\})$}
    \label{fig:memory-necessary}
    \end{subfigure}%
    \caption{
        MDP where memory and/or randomization are necessary for relational reachability properties with approximate equality.
    }
\end{figure}
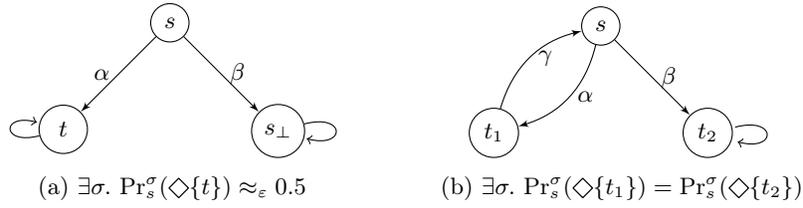

\begin{proof}
    We first show that randomization is necessary:
    Consider the MDP in \cref{fig:randomization-necessary}.
    Over general schedulers, the property 
    \enquote{$\exists \sched .\ \Pr^\sched_s(\Finally \{t\}) \approx_\epsilon 0.5$} holds for any $\epsilon$, but over (possibly memoryful) deterministic schedulers it holds only for $\epsilon \geq 0.5$.

    We further see that memory is also necessary, i.e., there are instances where memoryless randomized schedulers do not suffice:
    Consider the MDP in \cref{fig:memory-necessary}.
    Over general schedulers, the property
    \enquote{$\exists \sched_1 .\ \Pr^{\sched}_{s}(\Finally \{t_1\}) = \Pr^{\sched}_{s}(\Finally \{t_2\})$} holds, 
    but over memoryless schedulers it does not.
    Consider a memoryless randomized scheduler choosing $\alpha$ with probability $p$. If $p=1$ then the property does not hold. If $p<1$ then $\Pr^{\sched}_{s}(\Finally \{t_1\}) = p \neq \Pr^{\sched}_{s}(\Finally \{t_2\}) = \sum_{i=0}^{\infty} p^i (1-p) = \frac{1-p}{1-p} = 1$.
    \hfill\qed
\end{proof}
}{}

Note that \cref{th:MD_NP-complete,th:general_necessary} only make statements about properties with (approximate) equality.
Let us now consider \RelReach with inequality or disequality ($\comp~ \in \{\geq, >, \not\approx_\epsilon \mid \epsilon \geq 0 \}$).
Recall that, for inequality, we only have to check the maximizing (or, for $\not\approx_\epsilon$, possibly also the minimizing) schedulers for the goal unfoldings and transform them back to schedulers for the original MDP.
This transformation introduces memory in general:
If several probability operators are associated with the same scheduler but different initial states, then, intuitively, it might be necessary to switch behavior depending on the initial state.
Further, if several probability operators are evaluated under the same scheduler but for different target sets, then it might be necessary to switch behavior depending on which target sets were already visited. 

The next example illustrates that memory may be necessary for relational reachability properties with disequality even for just a single, absorbing target.

\begin{example}
    \label{ex:oneschedtwostate-exists-greater-memory}
    Consider the MDP depicted below and the property
    \[
        \exists \sched \in \Scheds .\ \Pr^{\sched}_{s_2}(\Finally \{ t \})  < \Pr^{\sched}_{s_1}(\Finally \{ t \}) \ .
    \]
    {\makeatletter
    \let\par\@@par
    \par\parshape0
    \everypar{}\begin{wrapfigure}{r}{0.4\textwidth}
    \vspace{-0.75cm}
    \scalebox{0.9}{
        \begin{tikzpicture}
            \node[state] (sinita) {$\statea$};   
            \node[state, right= of sinita] (sinitb) {$\stateb$};   
            
            \node[state, below left= of sinita, xshift=0.5cm] (t) {\phantom{$s_{2}$}}; 
            \node at (t) (th) {$t$};
            \node[state, below right= of sinita] (sx) {$s_\bot$}; 
    
            \node[dist, below left of=sinita, node distance=6ex] (alpha) {};
            \node[dist, below right of=sinita, node distance=6ex] (beta) {};
            \path[-, draw]
                (sinita) edge node[left] {$\alpha$} (alpha)
                (sinita) edge node[right] {$\beta$} (beta);
    
            \path[-latex', draw]
                (alpha) edge node[pos=0.25, left] {$1$} (t)
                (beta) edge node[pos=0.25, below] {$\frac{1}{2}$} (t)
                (beta) edge node[above] {$\frac{1}{2}$} (sx);
    
            \path[-latex', draw]
                (sinitb) edge (sinita)
                (t) edge[loop left] (t)
                (sx) edge[loop right] (sx);
        \end{tikzpicture}
    }
    \label{fig:oneschedtwostate-exists-greater-memory}
    \end{wrapfigure}
        Here, both probability operators are evaluated under the \emph{same scheduler} but \emph{different initial states}, and have the same, absorbing target set.
        Over MD schedulers this property cannot be satisfied: For all MD schedulers it holds that $\Pr^{\sched}_{s_1}(\Finally \{ t \}) = \Pr^{\sched}_{s_2}(\Finally \{ t \})$.
    \par}%
    In contrast, there does exist a memoryful scheduler such that the probability of reaching $t$ from $\statea$ exceeds the probability of reaching $t$ from $\stateb$, namely the scheduler that chooses $\alpha$ at $\statea$ if the execution was started at $\statea$, and $\beta$ otherwise.
    This also implies that 
    \[
        \exists \sched \in \Scheds .\ \Pr^{\sched}_{s_1}(\Finally \{ t \}) \neq \Pr^{\sched}_{s_2}(\Finally \{ t \})  \ 
    \]
    can be satisfied over general schedulers but not over MD schedulers.
\end{example}

\noindent However, under some mild restrictions on the target sets and/or the structure of the state-scheduler combinations, the procedure detailed in \cref{sec:algo} returns MD schedulers in \PTIME, if they exist. 
Correctness (\cref{thm:correctness}) then directly implies that MD schedulers suffice in these cases.
It remains future work to determine whether there are variants of \RelReachMD with disequality that are \NP-hard.

\begin{theorem}
    \label{th:MD_PTIME}
    For the following special cases of \RelReach with $\comp~ \in \{\geq, >, \not\approx_\epsilon \mid \epsilon \geq 0 \}$, the procedure detailed in \cref{sec:algo} runs in polynomial time and returns MD schedulers, if they exist:
    \begin{enumerate}[label=(\alph*)]
        \item \label{item:MD_independent}
        $n=m$.

        \item \label{item:MD_n-comb}
        Probability operators with the same scheduler variable have the same initial state 
        (formally, $\forall i,i' .\ \sched_{k_i} {=} \sched_{k_{i'}} \implies \state_i {=} \state_{i'}$)
        and all target sets are absorbing.

        \item \label{item:MD_sign}
        Probability operators with the same scheduler variable have equally signed coefficients and the same target sets
        (formally, $\forall i,i' .\ \sched_{k_i} {=} \sched_{k_{i'}} \implies ((q_i \geq 0 \iff q_{i'} \geq 0) \wedge T_i {=} T_{i'})$).
    \end{enumerate}
\end{theorem}

\begin{corollary}
    \label{th:MD_suffice}
    For the \RelReach variants from \cref{th:MD_PTIME}, MD schedulers suffice.
\end{corollary}

\iftoggle{extended}{
    Intuitively,%
}{
    The proof of \cref{th:MD_PTIME} can be found in \cite{extendedVersion}. 
    We focus here on the intuition.
    Firstly,%
}
if $n=m$ (\ref{item:MD_independent}),  
then all probability operators are independent in the sense that we only need to solve $n=m$ independent single-objective queries, for which there exist optimal MD schedulers~\cite{puterman1994markov}.\footnote{Note that this reasoning does not work for $\approx_\epsilon$, because there we may need to combine the optimal MD schedulers into memoryful randomized schedulers to obtain a witness. }

For \ref{item:MD_n-comb}, we consider the statement for $n=1$. 
If all probability operators are associated with the same scheduler and initial state and all target states are sinks, then
MD schedulers suffice as there is nothing to remember: We know which state we started from (since there is only a single initial state), and we know which target sets have already been visited (since all targets are sinks).   

Lastly, consider \ref{item:MD_sign} for $n=1$ and non-negative coefficients.
In this case,
since all target sets are the same, there must exist a MD scheduler maximizing all reachability probabilities $\Pr^\sched_{\state_i}(\Finally T)$ at the same time~\cite{puterman1994markov}.
Since all coefficients are non-negative, this scheduler also maximizes the weighted sum over the probabilities. We can analogously reason about minimizing the weighted sum.

\iftoggle{extended}{
    \begin{proof}[of \cref{th:MD_PTIME}]    
    For each case, we show that \cref{alg:linear_general} returns a MD scheduler and runs in polynomial time.
    
    \paragraph{\ref{item:MD_independent}:}
    If $n=m$ and $\comp \in \{\geq, >, \not\approx_\epsilon \mid \epsilon \geq 0 \}$,
    then $|\relInd(c)|=1$ for each $c \in \comb$ and each scheduler $\sched_i \in \{\sched_1, \ldots, \sched_n\}$ corresponds to a MD scheduler for $\indexc$ for the unique $c \in \comb$ with $c = (-, \sched_i)$.
    Further, any MD scheduler on $\indexc$ can be transformed back to a MD scheduler for $\mdp$ since there is only a single target set.
    
    We have already argued in \cref{th:general_PTIME} that \cref{alg:linear_general} runs in polynomial time in this case.

    \paragraph{\ref{item:MD_n-comb}:}
    We show the claim for the simplest case, namely $n=1$ and $s_1 = \ldots = s_m$ and all target sets are absorbing and $\comp \in \{\geq, >, \not\approx_\epsilon \mid \epsilon \geq 0 \}$.
    Then $|\comb| = 1$, so there is a unique $c \in \comb$.
    We can transform any MD scheduler for $\indexc$ back to a MD scheduler on $\mdp$ since the target sets are absorbing and there is a unique initial state $s_1 = \ldots = s_m$.
    
    Further, since target states are absorbing, the size of the goal unfolding $\indexc$ is linear in $\mdp$ and hence \cref{alg:linear_general} runs in polynomial time.
    
    The reasoning for the general case follows analogously.

    \paragraph{\ref{item:MD_sign}:}
    We show the claim for the simplest case, namely $n=1$ and $T_1 = \ldots = T_m =: T$ and for all $i=1, \ldots, m$ we have $q_i \geq 0$ and $\comp \in \{\geq, >, \not\approx_\epsilon \mid \epsilon \geq 0 \}$.
    We know that there exists some MD scheduler $\sched^{\max} \in \SchedsMD$ maximizing the reachability probability of $T$ for all $s \in \states$. 
    Since all coefficients $q_1, \ldots, q_m$ are non-negative, it follows that
    \begin{align*}
        \exists \sched \in \Scheds .\ \sum_{i=1}^{m} q_i \cdot \Pr^{\sched}_{\state_i} (\Finally T) \geq q_{m+1}
        \iff &
        \max_{\sched \in \Scheds} \sum_{i=1}^{m} q_i \cdot \Pr^{\sched}_{\state_i} (\Finally T) \geq q_{m+1}
        \\ \iff & 
        \sum_{i=1}^{m} q_i \cdot \Pr^{\sched^{\max}}_{\state_i} (\Finally T) \geq q_{m+1}
        \\ \iff & 
        \max_{\sched \in \SchedsMD} \sum_{i=1}^{m} q_i \cdot \Pr^{\sched}_{\state_i} (\Finally T) \geq q_{m+1}
    \end{align*}
    and analogously for $>$ and $\not\approx_\epsilon$.
    Hence, \cref{alg:linear_general} returns true iff the property is satisfiable over MD schedulers.
    
    Further, since $T_1 = \ldots = T_m$, the size of the goal unfolding is linear in the size of the original MDP and hence \cref{alg:linear_general} runs in polynomial time.
    
    The reasoning for the general case follows analogously.
    \hfill\qed
\end{proof}
}{}

\section{Implementation and Evaluation}
\label{sec:evaluation}

We have implemented a model checking algorithm based on the procedure described in \cref{sec:algo}, which we use to investigate what model sizes our approach can handle, and how it performs compared to existing tools that can check relational reachability properties.

\subsection{Setup}
Our prototype\footnote{Source code and benchmarks: \url{https://github.com/carolinager/RelReach}, artifact: \url{https://doi.org/10.5281/zenodo.15209574}}
 is implemented on top of (the python bindings of) the probabilistic model checker \storm~\cite{henselProbabilisticModel2022} and  supports an expressive fragment of relational reachability properties. It takes as input an MDP (in the form of a \PRISM file \cite{knp11}) with $m$ (not necessarily distinct) initial states and $m$ (not necessarily distinct) target sets, as well as a \emph{universally quantified} relational reachability property.
 Most available case studies for relational reachability properties are universally quantified.
Recall that, while the previous sections addressed existentially quantified relational reachability properties, we can transfer our results to universally quantified relational reachability properties by considering their negation.
If we conclude that the given (universal) property does not hold, out implementation can return witnesses in the form of schedulers returned by \storm for the expected rewards computed in Step 3, plus a solution for $\lambda$ as defined in the proof of \cref{thm:cruxApproxEqual} in case of (universal quantification with) $\not\approx$.
Scheduler export is disabled by default and also in our experiments.

Note that the tool’s output (`Yes'/`No') depends only on comparing an expected reward to a threshold (\iftoggle{extended}{\cref{alg:linear_general},~Lines~\ref{line:linear_general_return_first}--\ref{line:linear_general_return_last}}{\cref{alg:linear_general_simplified},~\cref{line:linear_general_simplified_return}}). 
Since the same calculations are performed regardless of the comparison's outcome, performance is generally unaffected by an instance's falsifiability.

We use \storm's internal single- and multi-objective model checking capabilities for the computation of expected rewards (Step 3 in \cref{sec:algo}), where single-objective model checking employs \emph{Optimistic Value Iteration}~\cite{hartmannsOptimisticValue2020} for approximate expected reward computation with a default tolerance of $10^{-6}$.\footnote{Note that \storm returns a single value $x$ with relative difference at most $10^{-6}$ to the exact result instead of sound lower and upper bounds here. We ensure soundness by computing conservative lower and upper bounds as $\tfrac{x}{1 + 10^{-6}}$ and $\tfrac{x}{1 - 10^{-6}}$.} 
The multi-objective model checking employs \emph{Sound Value Iteration}~\cite{DBLP:conf/cav/QuatmannK18} with the same tolerance.\footnote{\storm (currently) returns a single value with relative difference at most $10^{-6}$ to the exact result as both lower and upper bound and we again ensure soundness manually.}
We note that the implementation also supports \emph{exact} solving. 
Here, we report results for approximate computation of expected rewards with tolerance $10^{-6}$.
Our experiments were performed on a laptop with a single core of a $1.80$GHz Intel~i7 CPU and 16GB~RAM under Linux Ubuntu 24.04.1 LTS.

We compare our tool against two baselines:
\HyperProb\footnote{\url{https://github.com/TART-MSU/HyperProb}}~\cite{dabb21} and  \HyperPaynt\footnote{\url{https://github.com/probing-lab/HyperPAYNT}, we used this docker container: \url{https://zenodo.org/records/8116528}. Also referred to as \enquote{AR loop} in \cite{andriushchenkoDeductiveController2023}.}~\cite{andriushchenkoDeductiveController2023}. 
These tools handle (fragments of much more expressive) probabilistic hyperproperties that partially overlap with relational reachability properties. 
\HyperProb encodes the \HyperPCTL model-checking problem in SMT with an exponential number of variables.
\HyperPaynt uses abstraction refinement to model-check a fragment of \HyperPCTL, potentially exploring an exponential number of schedulers.
To the best of our knowledge, no other tools support a (nontrivial) fragment of relational reachability properties.
\HyperProb and \HyperPaynt support approximate comparison operators via an equivalent conjunction or disjunction of two inequalities.\footnote{While not covered in \cite{andriushchenkoDeductiveController2023}, \HyperPaynt also allows to add a constant to one side of the (in)equality.}

Both \HyperProb and \HyperPaynt search for policies and restrict themselves to MD schedulers. Solving a property over MD schedulers is not always equivalent to checking it over general schedulers (see, e.g., \cref{th:general_necessary}, \cref{ex:oneschedtwostate-exists-greater-memory}), but it coincides sometimes (\cref{th:MD_suffice}).
Further, \HyperProb supports only properties with a single scheduler quantifier (i.e., all probability operators are evaluated over the same MD scheduler).
\HyperPaynt, on the other hand, supports an arbitrary number of scheduler quantifiers.%
\footnote{For both \HyperProb and \HyperPaynt, one could alternatively consider a single scheduler on a manually created self-composition to simulate multiple schedulers (from the same initial state). Based on the large performance gap with \HyperProb and \HyperPaynt seen below, we did not consider such tweaks.}
To account for the different scheduler classes considered by the tools, we will in the following denote the validity of a property over general as well as over MD schedulers as \emph{HR result} and \emph{MD result}, respectively, when comparing the tools.%
\footnote{Note that \HyperPaynt expects existentially quantified properties, but we report the validity w.r.t.\ the universally quantified property, so the MD result stated here is the opposite of the result reported by \HyperPaynt.}
We will also denote for each benchmark family whether the considered property belongs to a fragment for which checking MD schedulers is equivalent to checking general schedulers.

Our tool's core procedure consists of polynomially many calls to well-established, practically efficient subroutines based on value iteration~\cite{QuatmannD0JK16}.
In contrast, \HyperProb and \HyperPaynt use exponential algorithms to solve the \NP-hard \RelReach problems over MD schedulers and do not optimize for the \PTIME special cases (\cref{th:general_PTIME}).
Further, to optimize a linear combination of expected rewards over different schedulers, we can treat the different state-scheduler combinations independently and then aggregate them in Step~4.

\subsection{Case Studies}

Let us first illustrate that \RelReach covers interesting problems.

\paragraph{Generalization of Von Neumann's trick (\VN).}
We generalize Von Neumann's trick from the motivating example of \cref{sec:intro} to $2N$ bits.
The idea is as follows: 
Extract the first $2N$ bits from the stream; if the number of zeros equals the number of ones, return the value of the first bit; otherwise try again.
We check whether
\[
    \forall \sched .\
    \Pr_{s_0}^{\sched}(\Finally \{\texttt{return 0}\}) \approx_{\epsilon} \Pr_{s_0}^{\sched}(\Finally \{\texttt{return 1}\})
\]
for $\epsilon=0$ and $\epsilon=0.1$ and varying values for $N$. 
By \cref{th:MD_suffice}\ref{item:absorbing}, checking this property over MD schedulers is equivalent to checking it over general schedulers since all target sets are absorbing.
Therefore, \HyperProb and \HyperPaynt can also check the general problem in this case even though they restrict to MD schedulers.

\begin{table}[t] 
    \centering
    \caption{
        Experimental results.
        \emph{HR res.}/\emph{MD res.}: Does the universally quantified property hold over general (HR)/MD schedulers? 
        \emph{$\equiv?$}: Is checking the property over general scheduler equivalent to checking over MD schedulers?
        \emph{Time} is rounded to the nearest second, with the exception of values ${<}1s$ which are denoted as such.
        For every tool, we give the total time and, in brackets,
        (1) for our tool: the total time excluding model building,  
        (2) for \HyperProb: z3 solving time, 
        (3) for \HyperPaynt: the reported `synthesis time'. 
        \TO: Timeout (1 hour). \OOM: Out of memory.
    }
    \label{tab:eval_new}
    \adjustbox{max width=\textwidth}{  
    \begin{tabular}{l l l r || c | r >{\color{gray}}r || c || c | r >{\color{gray}}r | r >{\color{gray}}r}
        \toprule
         \multicolumn{4}{c||}{Case study} & \multicolumn{3}{c||}{Our tool} & \multirow{2}{*}{$\equiv$?} & \multicolumn{5}{c}{Comparison \textbf{over MD sched.}} \\
         & \multicolumn{2}{c}{Variant} & \makecell{$|\states|$} & HR res. & \multicolumn{2}{c||}{Time} & &
         MD res. & 
         \multicolumn{2}{c|}{\HyperProb} & \multicolumn{2}{c}{\HyperPaynt}
         \\
         \midrule 
         \multirow{8}{*}{\VN} 
         & $N{=}1$ & $\epsilon{=}0$ & 5 
         & No & ${<}1s$ & (${<}1s$)
         & \multirow{12}{*}{\makecell{$\equiv$ \\ Cor.~\ref{th:MD_suffice}}}
         &  No
         & ${<}1s$ & (${<}1s$)
         &  ${<}1s$ & (${<}1s$)
         \\
         & $N{=}1$ & $\epsilon{=}0.1$ & 5 
         & Yes & ${<}1s$ & (${<}1s$)
         &
         & Yes
         & ${<}1s$ & (${<}1s$)
         & ${<}1s$ & (${<}1s$)
         \\
         & $N{=}10$ & $\epsilon{=}0$ & 383 
         & No & ${<}1s$ & (${<}1s$)
         & & No
         & \TO & -
         & ${<}1s$ & (${<}1s$)
         \\
         & $N{=}10$ & $\epsilon{=}0.1$ & 383 
         & No & ${<}1s$ & (${<}1s$)
         & & No
         & \TO & -
         & ${<}1s$ & (${<}1s$)
         \\
         & $N{=}100$ & $\epsilon{=}0$ & $39\,803$ 
         & No & $3s$ & ($3s$)
         & & -
         & \TO & -
         & \TO & -
         \\
         & $N{=}100$ & $\epsilon{=}0.1$ & $39\,803$ 
         & No & $3s$ & ($3s$)
         & & -
         & \TO & -
         & \TO & -
         \\
         & $N{=}200$ & $\epsilon{=}0$ & $159\,603$
         & No & $218s$ & ($217s$)
         & & -
         & \OOM & -
         & \TO  & -
         \\
         & $N{=}200$ & $\epsilon{=}0.1$ & $159\,603$
         & No & $209s$ & ($209s$)
         & & -
         & \OOM & -
         & \TO & -
         \\
         & $N{=}250$ & $\epsilon{=}0$ & $249\,503$
         & No & $1\,358s$ & ($1\,357s$)
         & & -
         & \OOM & -
         & \TO  & -
         \\
         & $N{=}250$ & $\epsilon{=}0.1$ & $249\,503$
         & No & $1\,323s$ & ($1\,322s$) 
         & & -
         & \OOM & -
         & \TO & -
         \\
         & $N{=}300$ & $\epsilon{=}0$ & $359\,403$
         & - & \TO & -
         & & -
         & \OOM & -
         & \TO  & -
         \\
         & $N{=}300$ & $\epsilon{=}0.1$ & $359\,403$
         & - & \TO & -  
         & & -
         & \OOM & -
         & \TO & -
         \\
         \midrule
         \multirow{8}{*}{\RT} 
         & $N{=}10$ & $(N,N)$ & $933$ 
         & Yes & ${<}1s$ & (${<}1s$)
         & \multirow{8}{*}{\makecell{$\equiv$ \\ Cor.~\ref{th:MD_suffice}}}
         & Yes
         & \multicolumn{2}{c|}{\multirow{8}{*}{n/a}}
         & ${<}1s$ & (${<}1s$)
         \\
         & $N{=}10$ & $(N,N{-}1)$ & $1\,021$ 
         & No & ${<}1s$ & (${<}1s$)
         & & No
         & & 
         & ${<}1s$ & (${<}1s$)
         \\
         & $N{=}100$ & $(N,N)$ & $994\,803$ 
         & Yes & $6s$ & (${<}1s$)
         & & -
         & & 
         & \TO & -
         \\
         & $N{=}100$ & $(N,N{-}1)$ & $1\,004\,701$ 
         & No & $5s$ & (${<}1s$)
         & & -
         & & 
         & \TO & -
         \\
         & $N{=}200$ & $(N,N)$ & $7\,979\,603$
         & Yes & $48s$ & ($10s$)
         & & -
         & &
         & \TO & - 
         \\
         & $N{=}200$ & $(N,N{-}1)$ & $8\,019\,401$
         & No & $44s$ & ($6s$)
         & & -
         & &
         & \TO & -
         \\
         & $N{=}300$ & $(N,N)$ & $26\,954\,403$
         & - & \OOM & -
         & & -
         & &
         & \OOM & -
         \\
         & $N{=}300$ & $(N,N{-}1)$ & $27\,044\,101$
         & - & \OOM & -
         & & -
         & &
         & \OOM & -
         \\
         \bottomrule
    \end{tabular}
    }
\end{table}

\paragraph{Robot tag (\RT).}
Consider a $N {\times} N$ grid world with a robot and a janitor, where both move in turns, starting with the robot.%
\footnote{Our model is based on a \href{https://www.prismmodelchecker.org/casestudies/robot.php}{\PRISM model} and a \href{https://github.com/sjunges/gridworld-by-storm/blob/master/gridstorm/models/files/evade-two-mdp.nm}{Gridworld-By-Storm model}.}
The robot starts in the lower left corner, the janitor in the upper right corner.
The robot has fixed the following strategy: It moves right until it reaches the lower right corner, then moves up until it reaches the upper right corner, its target.
The janitor can hinder the robot from reaching its target by occupying a cell that the robot wants to move to. 
We now want to check whether this strategy is approximately robust against adversarial behavior by the janitor, 
in the sense that the probability of reaching the target should be approximately independent of the janitor strategy.
Formally,
\[
    \forall \sched_1 , \sched_2 .\ \Pr^{\sched_1}_s(\Finally \{t\}) \approx_{10^{-5}} \Pr^{\sched_2}_s(\Finally \{t\})
    ~,
\]
where $s$ represents the initial state with robot and janitor in their initial locations, and
$t$ represents that the robot has reached its target.
We check this property for different grid sizes $N$ and different starting positions for the janitor:
If the janitor starts in the upper right corner, they cannot hinder the robot, but if they start in location $(N, N{-}1)$, then they can always stop the robot. 

By \cref{th:MD_suffice}\ref{item:independent}, MD schedulers suffice here. 
However, \HyperProb does not natively support this problem as it only allows a single scheduler quantifier and the property is trivially satisfied if we use a single scheduler quantifier for both probability operators.

\paragraph{Results.}
The results for both case studies are presented in \cref{tab:eval_new}.
Firstly, we observe that the quality of the coin simulation decreases with a growing number of bits $N$: While approximate equality with $\epsilon{=}0.1$ holds for $N{=}10$, it does not hold anymore for $N{=}100$.
Further, our tool handles 1 million states in $5$--$6s$ for \RT but times out for the \VN instances of comparable size.
Besides the size of the state space, the computation time of expected rewards also depends on structural aspects, e.g. many cycles (as in \VN) cause slower convergence of the approximation algorithm for expected rewards.
On problem instances that are also supported by \HyperProb or \HyperPaynt, our tool performs drastically better: \HyperProb times out already for \VN with $N{=}10$, and \HyperPaynt for \VN with $N{=}100$ while our tool solves these instances in a matter of seconds.

\subsection{Benchmarks for Probabilistic Hyperproperties}
\label{sec:case-studies}

Next, we investigate the scalability of our tool on benchmarks from the literature on probabilistic hyperproperties that are relational reachability properties. These benchmarks are typically motivated by security use cases.
In particular, three out of four \HyperPCTL case studies for MDPs presented in \cite{abbd20-atva} are covered by our approach: 
\iftoggle{extended}{%
    We consider (mild variations of) \TA, \PW, \TS from \cite{abbd20-atva}. 
    Further, we consider \SD from \cite{andriushchenkoDeductiveController2023}.
    \label{app:case-studies}

\paragraph{\TA}
    A standard application of probabilistic hyperproperties
    is to check whether an implementation of modular exponentiation for RSA public-key encryption 
    has a side-channel \emph{timing leak} \cite{abbd20-atva}. 
    Concretely, the setting is as follows: One thread computes $a^b \text{ mod } n$ for a given plaintext $a$, encryption key $b$, and modulus $n$ (all non-negative integers), 
    while a second (attacker) thread tries to infer the value of $b$ by keeping a counter $c$ to measure the time taken by the first thread.
    \cref{fig:modexp} shows the implementation of the modular exponentiation algorithm that we want to check for side-channels.

    If $b$ has $k$ bits, we can model the parallel execution of the two threads as an MDP with $2^k$ initial states $\Init$, each representing a different value of $b$.
    In our models, a (memoryful randomized) scheduler for the MDP corresponds to a thread scheduler.
    In contrast, in \cite{abbd20-lpar,andriushchenkoDeductiveController2023} the authors hard-coded fair thread schedulers in their models, and a (memoryless deterministic) scheduler for the MDP corresponds to a secret input. 

    The goal is to now check whether the implementation satisfies \emph{scheduler-specific probabilistic observational determinism} (SSPOD) \cite{nsh13}, i.e., whether for any scheduling of the two threads that no information about the secret input (here, $b$) can be inferred by the publicly observable information (here, the time taken by the modular exponentiation thread).
    
    In \cite{abbd20-atva}, the desired property is formulated in \HyperPCTL as follows:
    \[\forall \varsched_1 \forall \varsched_2 .\ \forall \varstate_1(\varsched_1) \forall \varstate_2(\varsched_2) .\ 
    (\init_{\varstate_1} \wedge \init_{\varstate_2}) 
    \implies 
    \bigwedge_{j=0}^{2k} \Prob(\Finally (c=j)_{\varstate_1}) = \Prob(\Finally (c=j)_{\varstate_2})\]
    where 
    $\init$ marks initial states for different values of the secret input $b$, 
    $b$ has $k$ bits, and 
    $c$ is the counter of the attacking side-channel thread.
    An MDP $\mdp$ satisfies the above \HyperPCTL formula iff it holds that
    \begin{align*}
        &\forall \sched_1, \sched_2 \in \Scheds .\ \bigwedge_{\state_1,\state_2 \in \Init} \bigwedge_{j=0}^{2k} \Pr^{\sched_1}_{\state_1}(\Finally (c=j)) = \Pr^{\sched_2}_{\state_2}(\Finally (c=j))
        ~.
    \end{align*}
    Since universal quantification distributes over conjunction, we can reformulate this to a conjunction of universally quantified relational reachability properties as follows:
    \begin{align*}
        \bigwedge_{\state_1,\state_2 \in \Init} \bigwedge_{j=0}^{2k} \forall \sched_1, \sched_2 .\ \Pr^{\sched_1}_{\state_1}(\Finally (c=j)) = \Pr^{\sched_2}_{\state_2}(\Finally (c=j))
        ~.
    \end{align*}

    In our experiments, we check the first conjunct of this property, i.e., we compare initial values $b=0$ and $b=1$ for $j=0$.
    We check the property for $M \in \{8,16,24,28,32\}$ where $M=2k$ where $k$ is the number of bits for the secret input $b$. 

\begin{figure}[t]
    \centering
    \begin{subfigure}[t]{0.45\textwidth}
    \centering
    \begin{lstlisting}[style=CStyle,tabsize=2,language=ML,basicstyle=\scriptsize,escapechar=/,backgroundcolor=\color{white},]
int mexp(a,b,n){
  d = 0; e = 1; i = k;
  while (i >= 0){
    i = i-1; d = d*2;
    e = (e*e) % n;
    if (b(i) = 1)
      d = d+1;
      e = (e*a) % n;
    }
} \end{lstlisting}
    \end{subfigure}%
    \quad
    \begin{subfigure}[t]{0.45\textwidth}
    \centering
    \begin{lstlisting}[style=CStyle,tabsize=2,language=ML,basicstyle=\scriptsize,escapechar=/,backgroundcolor=\color{white},]
  t = new Thread(mexp(a,b,n)); 
  c = 0; M = 2 * k;
  while (c < M & !t.stop) c++;   \end{lstlisting}
    \end{subfigure}
    \caption{Modular exponentiation (left) and attacker thread (right) \cite{abbd20-lpar}.
    } 
    \label{fig:modexp}
\end{figure}

\paragraph{\PW} 
    Another example for an implementation with possible timing leaks is a careless implementation of string comparison for password verification where a timing leak may reveal information about the password \cite{abbd20-lpar}.
    As for \TA, we want check whether an attacker thread observing the time taken by the main thread (executing a string comparison algorithm) can infer information about the password.

    \cref{fig:string} displays the string comparison algorithm that we want to check for information-leaks.
\begin{figure}[t]
    \centering
    \begin{subfigure}[t]{0.5\textwidth}
    \centering
    \begin{lstlisting}[style=CStyle,tabsize=2,language=ML,basicstyle=\scriptsize, backgroundcolor=\color{white},
    %    ,escapechar=/
      ]
int str_cmp(char * r){
  char * s = 'Bg\$4\0';
  i = 0;
  while (s[i] != '\0'){
    i++;
    if (s[i]!=r[i]) return 0;
    }
    return 1;
}    \end{lstlisting}
    \end{subfigure}
    \caption{String comparison \cite{abbd20-lpar}.}
    \label{fig:string}
\end{figure}

    We model this by mapping strings to integers and allowing $2^{8n}$ different input strings for a string length of $n$.
    As for \TA, our model differs from \cite{abbd20-lpar,andriushchenkoDeductiveController2023}: In our model, (memoryful randomized) MDP schedulers correspond to thread schedulers, while they hard-coded fair thread schedulers in their models, and a (memoryless deterministic) scheduler for the MDP corresponds to a secret input. 

    We check the same property as for \TA:
    In our experiments, we compare string inputs represented by the integers 0 and 1, for $j=0$.
    We check the property for $M \in \{2,4\}$ where $M=2n$ where $n$ is the length of the string.

\paragraph{\TS}
    Consider the following classic example of a simple insecure multi-threaded program~\cite{smith03} 
    \[ 
    th \colon \text{\bf{while }} h>0 \text{\bf{ do }} \{h \leftarrow h-1\};\; l \leftarrow 2 
    \quad \parallel \quad
    th' \colon l \leftarrow 1  
    \]
    where $h$ is a secret input and $l$ a public output.
    Intuitively, this program is not secure because the final value of $l$ allows to make a probabilistic inference on the initial value of the secret input $h$.
    For example, under a fair scheduler (that schedules all available threads with the same probability), the probability of outputting $l=2$ is much higher than the probability of outputting $l=1$.
    Formally, the program violates \emph{scheduler-specific probabilistic observational determinism} (SSPOD)~\cite{nsh13}.
    SSPOD stipulates that for all schedulings of the threads 
    the probability of observing $l=1$ in the end should be the same as the probability of observing $l=2$. 
    \cite{abbd20-atva} formulates the desired property in \HyperPCTL as follows:
    \[\forall \varsched .\ \forall \varstate_1(\varsched) \forall \varstate_2(\varsched) .\ 
    (\init_{\varstate_1} \wedge \init_{\varstate_2}) 
    \implies 
    \bigwedge_{j=1}^{2} \Prob(\Finally (l=j)_{\varstate_1}) = \Prob(\Finally (l=j)_{\varstate_2})\]
    where 
    $\init$ marks initial states for different values of the secret input $h$.
    
    This is equivalent to the following conjunction of universally quantified relational reachability properties:
    \begin{align*}
        \bigwedge_{s_1, s_2 \in \Init} \bigwedge_{j=1}^{2} \forall \sched .\ \Pr^{\sched}_{\state_1}(\Finally (l=j)) = \Pr^{\sched}_{\state_2}(\Finally (l=j))
        ~.
    \end{align*}
    where $\Init$ is the set of states representing initial states of the program for different values of $h$ (up to a certain bound).

    We compare the initial values of $h$ indicated in the table, for $j=1$.
    Note that this is actually equivalent to checking the full conjunction since $\Pr^\sched_s(\protect\Finally (l=1)) = 1 - \Pr^\sched_s(\protect\Finally (l=2))$.

\paragraph{\SD}    
    \cite{andriushchenkoDeductiveController2023} extend the notion of \emph{stochastic domination} \cite{bartheProbabilisticCouplings2020} to MDP states by defining that a state $s_1$ stochastically dominates a state $s_2$ w.r.t.\ a target set $T$ iff 
    \[ 
        \forall \sched .\ \Pr^{\sched}_{s_1}(\Finally T) \geq \Pr^{\sched}_{s_2}(\Finally T) ~,
    \]
    which is a natural universally quantified relational reachability property.
    In \cite{andriushchenkoDeductiveController2023}, this property is applied to various robot-maze problems with a fixed pair of initial locations each, checking whether the first location can guarantee a better reachability probability than the other location, no matter how the robot behaves.
    We check the same property (over memoryful randomized schedulers) on all mazes provided in \cite{andriushchenkoDeductiveController2023}.

}{%
    We consider (mild variations of) \TA, \PW, \TS from \cite{abbd20-atva}. 
    \TA and \PW check properties of the form
    \begin{align*}
        \textstyle\bigwedge_{\state_1,\state_2 \in \Init} \bigwedge_{i=0}^{M} \forall \sched_1, \sched_2 .\ \Pr^{\sched_1}_{\state_1}(\Finally T_i) = \Pr^{\sched_2}_{\state_2}(\Finally T_i)
        ~,
    \end{align*}
    where $\Init$ is a set of initial states.
    Here, we benchmark only the first conjunct (which can be falsified) for all tools.
    We consider two variations of \TA and \PW: \TAone and \PWone use only a single scheduler quantifier for both probability operators while \TAtwo and \PWtwo use two scheduler quantifiers, as in the original formulation~\cite{abbd20-atva}.\footnote{Note that both variants are equivalent over general schedulers.}
    The property for \TS is analogous, but with only a single scheduler quantifier and we fix a different pair of initial states for each instance. 
    Further, we consider \SD from \cite{andriushchenkoDeductiveController2023}, which checks 
    \[ 
        \forall \sched .\ \Pr^{\sched}_{s_1}(\Finally T) \geq \Pr^{\sched}_{s_2}(\Finally T) ~.
    \]
    Details on all four benchmarks can be found in \cite{extendedVersion}, including the differences in our models to the models from \cite{abbd20-atva,andriushchenkoDeductiveController2023}.
}

\begin{table}[t] 
    \centering
    \caption{
        Comparison on benchmarks for probabilistic hyperproperties.
        The meaning of columns and abbreviations is the same as in \cref{tab:eval_new}.
    }
    \label{tab:eval}
    \adjustbox{max width=\textwidth}{  
    \begin{tabular}{l l r || c | r >{\color{gray}}r || c || c | r >{\color{gray}}r | r >{\color{gray}}r}
        \toprule
         \multicolumn{3}{c||}{Case study} & \multicolumn{3}{c||}{Our tool} & \multirow{2}{*}{$\equiv$?} & \multicolumn{5}{c}{Comparison \textbf{over MD sched.}} \\
         & Variant & \makecell{$|\states|$} & HR res. & \multicolumn{2}{c||}{Time} & &
         MD res. & 
         \multicolumn{2}{c|}{\HyperProb} & \multicolumn{2}{c}{\HyperPaynt} \\
         \midrule 
         \multirow{5}{*}{\TAone} 
         & $M{=}8$ & $423$
         & No & ${<}1s$  & (${<}1s$)  
         & \multirow{5}{*}{\makecell{$\not\equiv$ \\ Ex~\ref{ex:oneschedtwostate-exists-greater-memory}}}
         & No 
         & $3\,055s$ & ($2\,858s$) 
         & ${<}1s$  & (${<}1s$)  
         \\
         & $M{=}16$ & $13\,039$ 
         & No & ${<}1s$  & (${<}1s$) 
         & & No 
         & \OOM & - 
         & $21s$ &  ($1s$)
         \\
         & $M{=}24$ & $307\,175$ 
         & No & $20s$ & (${<}1s$)
         & & - 
         & \OOM & - 
         & \TO & -
         \\
         & $M{=}28$ & $1\,425\,379$ 
         & No & $490s$ & (${<}1s$)
         &  & -
         & \OOM & - 
         & \TO & -
         \\
         & $M{=}32$ & $6\,488\,031$
         & -  & \TO & -
         & & -
         & \TO &
         & \TO & -
         \\
         \midrule
         \multirow{5}{*}{\TAtwo} 
         & $M{=}8$ & 423 
         & No & ${<}1s$ & (${<}1s$) 
         &\multirow{5}{*}{\makecell{$\equiv$ \\ Cor.~\ref{th:MD_suffice}}}
         & No 
         & \multicolumn{2}{c|}{\multirow{5}{*}{n/a}}  
         & ${<}1s$ & (${<}1s$)
         \\
         & $M{=}16$ & $13\,039$ 
         & No & ${<}1s$ & (${<}1s$) 
         & & No 
         &  &  
         & $115s$ & ($2s$)
         \\
         & $M{=}24$ & $307\,175$ 
         & No & $20s$ & (${<}1s$) 
         & & - 
         &  &  
         & \TO & - 
         \\
         & $M{=}28$ & $1\,425\,379$ 
         & No & $493s$ & (${<}1s$)
         & & -
         & & 
         & \TO & -
         \\
         & $M{=}32$ &  $6\,488\,031$
         & - & \TO & -
         & & -
         &  &
         & \TO & -
         \\
         \midrule
         \multirow{2}{*}{\PWone} 
         & $M{=}2$ & $2\,307$ 
         & No & ${<}1s$ & (${<}1s$) 
         & \multirow{2}{*}{ \makecell{$\not\equiv$ \\ Ex~\ref{ex:oneschedtwostate-exists-greater-memory}} }
         & No 
         & \OOM & - 
         & ${<}1s$ & (${<}1s$) 
         \\
         & $M{=}4$ & $985\,605$ 
         & - & \TO & - 
         &  & -
         & \TO & - 
         & \TO & -
         \\
         \midrule
         \multirow{2}{*}{\PWtwo} 
         & $M{=}2$ & $2\,307$ 
         & No & ${<}1s$ & (${<}1s$)  
         & \multirow{2}{*}{ \makecell{$\equiv$ \\ Cor.~\ref{th:MD_suffice}} }
         & No 
         & \multicolumn{2}{c|}{\multirow{2}{*}{n/a}} 
         & $5s$ & (${<}1s$)
         \\
         & $M{=}4$ & $985\,605$ 
         & - & \TO & - 
         & & - 
         & &  
         & \TO  & -
         \\
         \midrule
         \multirow{5}{*}{\TS} 
         & $h{=}(10,20)$ & 252 
         & No & ${<}1s$ & (${<}1s$) 
         & \multirow{5}{*}{ \makecell{$\not\equiv$ \\ Ex~\ref{ex:oneschedtwostate-exists-greater-memory}} }
         & No 
         & $112s$ & ($38s$)
         & ${<}1s$ & (${<}1s$)
         \\
         & $h{=}(20,200)$ & $2\,412$ 
         & No & ${<}1s$ & (${<}1s$) 
         &  & No
         & \OOM & - 
         & $2s$ & (${<}1s$)
         \\
         & $h{=}(20,5\,000)$ & $60\,012$ 
         & No & ${<}1s$ & (${<}1s$) 
         &  & No
         & \OOM & -
         & $1\,076s$ & ($20s$)
         \\
         & $h{=}(50,10\,000)$ & $120\,012$ 
         & No & ${<}1s$ & (${<}1s$) 
         &  & -
         & \OOM & -
         & \TO & -
         \\
         & $h{=}(50,20\,000)$ & $240\,012$ 
         & No & ${<}1s$ & (${<}1s$) 
         &  & 
         & \OOM & -
         & \TO & -
         \\
         \midrule 
         \multirow{7}{*}{\SD} 
         & simple & 10 
         & No & ${<}1s$ & (${<}1s$)  
         & \multirow{7}{*}{ \makecell{$\not\equiv$ \\ Ex~\ref{ex:oneschedtwostate-exists-greater-memory}} }
         & Yes 
         & $3s$  & (${<}1s$) 
         & $2s$ & ($2s$)
         \\
         & splash-1 & 16 
         & No & ${<}1s$ & (${<}1s$) 
         &  & No 
         & $1\,184s$  & ($1\,183s$) 
         & ${<}1s$ & (${<}1s$) 
         \\
         & splash-2 & 25 
         & No & ${<}1s$ & (${<}1s$) 
         & & Yes 
         & \TO & - 
         & $2\,487s$ & ($2\,485s$)
         \\
         & larger-1 & 25 
         & No & ${<}1s$ & (${<}1s$) 
         & & No 
         & \TO  & - 
         & $310s$ & ($310s$)
         \\
         & larger-2 & 25 
         & No & ${<}1s$ & (${<}1s$) 
         & & Yes 
         & \TO & - 
         & $925s$ & (${<}1s$)
         \\
         & larger-3 & 25 
         & No & ${<}1s$ & (${<}1s$)
         & & No 
         & \TO & - 
         & ${<}1s$ & (${<}1s$)
         \\
         & train & 48 
         & No & ${<}1s$ & (${<}1s$)
         & & Yes 
         & \TO & - 
         & $14s$ & ($14s$)
         \\
         \bottomrule
    \end{tabular}
    }
\end{table}

\paragraph{Results.}
\cref{tab:eval} presents our 
experimental 
results on these four case studies and provides a comparison with \HyperProb and \HyperPaynt.
We observe that for all benchmarks in \cref{tab:eval}, the running time of our tool consists almost entirely of building the model.

For the benchmarks with a single scheduler quantifier (\TAone, \PWone, \TS, \SD), checking over general schedulers is, in general, not equivalent to checking over MD schedulers (see \cref{ex:oneschedtwostate-exists-greater-memory}). 
Notably, already instances with $100\,000$ states prove to be challenging over MD schedulers for \HyperProb and \HyperPaynt, while our tool solves these instances over general schedulers in over a minute and can handle instances of \TAone with 1 million states in less than 10 minutes.

For the benchmarks with two scheduler quantifiers (\TAtwo and \PWtwo), 
checking over MD schedulers is equivalent to checking over general schedulers since both probability operators are independent (\cref{th:MD_suffice}).%
\footnote{Checking for two scheduler quantifiers on a manually created self-composition of the MDP with \HyperProb already exceeds memory bounds for \TA with $m{=}8$.}
Notably, our tool solves the instance for $M{=}24$ in over a minute while \HyperPaynt times out.
Our tool can solve instances of \TAtwo with more than 1 million states, but times out for the instance of \PWtwo with almost 1 million states.

In summary, our tool is orders of magnitude faster than existing tools, which restrict to MD schedulers but nevertheless solve an equivalent problem in some cases (\TAtwo, \PWtwo).
For problem instances belonging to fragments whose decision problems are \NP-hard over MD schedulers but in \PTIME over general schedulers, we show that solving the \NP-hard problem via SMT solving (\HyperProb) or an abstraction-refinement approach (\HyperPaynt) is also much harder in practice than solving the \PTIME problem.

\section{Related Work}
\label{sec:related}

We discuss two main areas of related work: \emph{multi-objective MDP model checking} and \emph{probabilistic hyperlogics}.

\paragraph{Multi-objective model checking.}
The techniques and results of this paper are strongly related to multi-objective model checking (MOMC) for 
MDPs~\cite{chatterjeeMarkovDecision2006,chatterjeeMarkovDecision2007,etessamiMultiObjectiveModel2008,quatmannVerificationMultiobjective2023}. In 
MOMC, the key question is: \emph{Is there a scheduler such that target set $A$ is reached with probability \underline{at least} $\lambda_A$ \underline{and} target set $B$ is reached 
with probability \underline{at least} $\lambda_B$?}
With appropriate preprocessing, MOMC can moreover ask whether a set 
of states is reached with probability \emph{exactly}, say, $0.5$ --- but not whether there exists a scheduler such that $A$ and $B$ are reached with equal (or approximately equal) probability.
Indeed, MOMC and the relational properties we study are incomparable 
in the sense that they are disjoint up to single-objective queries:
MOMC focuses on \emph{conjuctive} queries but each conjunct can only compare a probability to a \emph{constant}.
In contrast, our relational properties are single constraints involving multiple probability operators; they do not support conjunctions but allow \emph{comparing} probabilities.
Similar to our results, however, MOMC with a fixed number of targets  is \NP-complete over MD schedulers~\cite{quatmannVerificationMultiobjective2023}, and solvable in \PTIME over general schedulers.
Moreover, prominent algorithms~\cite{forejtParetoCurves2012} for MOMC also rely on optimizing weighted sums of probabilities just like we do.
MOMC is supported in probabilistic model checkers like \PRISM~\cite{knp11} and \storm~\cite{henselProbabilisticModel2022}.
Due to the similarities outlined above, we were able to reuse some of \storm's MOMC capabilities in our implementation.

\paragraph{Probabilistic hyperlogics.}
Our paper was motivated by recent emerging interest in algorithms for probabilistic 
hyperlogics, like \HyperPCTL~\cite{ab18,abbd20-atva} and 
\PHL~\cite{dft20}.
Relational reachability properties are a strict fragment of \HyperPCTL. 
While \PHL cannot naturally compare probabilities from different initial states, every MDP $\mdp$ and \RelReach property $\varphi$ can be transformed to an MDP $\mdp'$ and \PHL formula $\varphi'$ s.t.\ $\mdp$ satisfies $\varphi$ iff $\mdp'$ satisfies $\varphi'$ over general schedulers (by making a copy of the MDP for every state-scheduler combination).
\HyperPCTL and \PHL both can express properties like \emph{Does there exist a scheduler such that all paths are trace-equivalent almost-surely?}, which \RelReach does not cover.
Due to their high expressiveness, the corresponding model-checking problems are undecidable. 
This paper contributes two main points to the study of probabilistic hyperlogics.
First, searching for randomized and memoryful schedulers may be beneficial complexity-wise, both in theory and in practice.
Second, many of the motivating case studies for probabilistic hyperlogics can also be treated by dedicated and therefore much more efficient routines, which also motivated the use of an AR-loop in~\cite{andriushchenkoDeductiveController2023}. 
However, in \cite{andriushchenkoDeductiveController2023} a \emph{search} over a finite amount of (MD) schedulers is suggested, while we study the computational complexity and consider an algorithm for general (and thus uncountably many) schedulers.
While \uppaalsmc is not a tool for probabilistic hyperproperties, it also supports \emph{statistical} model-checking for comparison of two cost-bounded reachability  probabilities \emph{on DTMCs}~\cite{davidTimeStatistical2011}. 
\section{Conclusion}
\label{sec:conclusion}

In this paper, we investigated probabilistic model checking for relational reachability properties in MDPs. 
We presented a practically and theoretically efficient algorithm for  these properties and demonstrated the hardness of some other fragments. 
Relational reachability properties are a sub-class of probabilistic hyperproperties and we showed that, compared to existing approaches for model checking probabilistic hyperproperties, our approach does not (need to) restrict itself to only memoryless deterministic schedulers. In addition, it scales orders of magnitudes better on various benchmarks used to motivate probabilistic hyperproperties.

For future work, we aim to study more expressive fragments. In particular, we would like to consider Boolean combinations (generalizing to multi-objective model checking), relate $\omega$-regular properties or properties that require trace equivalence, and consider models like POMDPs with restricted scheduler classes.

\begin{credits}
\subsubsection{\ackname} 
We thank Oyendrila Dobe for fruitful discussions on earlier versions of this work.
We are grateful to Tim Quatmann for his continuous and reliable \storm support.
Lina Gerlach and Tobias Winkler are supported by the DFG RTG 2236/2 \textit{UnRAVeL}. 
Sebastian Junges is supported by the NWO VENI Grant ProMiSe (222.147).
This work is partially sponsored by the United States National Science Foundation (NSF) Award SaTC 2245114.

\subsubsection{\discintname}
The authors have no competing interests to declare that are relevant to the content of this article. 
\end{credits}

\bibliographystyle{splncs04}
\bibliography{bibliography}

\end{document}